\pdfoutput=1
\documentclass[%
 reprint,
 amsmath,amssymb,
 aps,
 nofootinbib,
 superscriptaddress,
 longbibliography,
 showkeys
]{revtex4-2}

\usepackage{amsthm}
\usepackage{graphicx}
\usepackage{xcolor}
\usepackage{braket}
\usepackage{amsmath}
\usepackage{amssymb}
\usepackage{enumerate}
\usepackage{mathtools}
\usepackage{etoolbox}
\usepackage{algorithm}
\usepackage{algpseudocode}
\usepackage{bbm}
\usepackage{float}
\usepackage[
colorlinks=true,
urlcolor=blue,
citecolor=blue,
linkcolor=blue,
hyperfootnotes=false,
breaklinks=true
]{hyperref}
\usepackage{subfigure}
\usepackage{caption}
\usepackage{subcaption}
\usepackage{tikz}
\usetikzlibrary{quantikz}

\newtheorem{theorem}{Theorem}

\newtheorem{definition}{Definition}

\newtheorem*{theorem1restate}{Theorem \ref{th:LegExAccu}}
\newtheorem*{theoremmainrestate}{Theorem \ref{th:main}}

\newcommand{\proc}[1]{\textup{\textsf{#1}}}



\begin{document}

\title{Dividing quantum circuits for time evolution of stochastic processes by orthogonal series density estimation}
\author{Koichi Miyamoto}
\email{miyamoto.kouichi.qiqb@osaka-u.ac.jp}
\affiliation{Center for Quantum Information and Quantum Biology, The University of Osaka, Toyonaka, Osaka, Japan}

\date{\today}

\begin{abstract}

Quantum Monte Carlo integration (QMCI) is a quantum algorithm to estimate expectations of random variables, with applications in various industrial fields such as financial derivative pricing.
When QMCI is applied to expectations concerning a stochastic process $X(t)$, e.g., an underlying asset price in derivative pricing, the quantum circuit $U_{X(t)}$ to generate the quantum state encoding the probability density of $X(t)$ can have a large depth.
With time discretized into $N$ points, using state preparation oracles for the transition probabilities of $X(t)$, the state preparation for $X(t)$ results in a depth of $O(N)$, which may be problematic for large $N$.
Moreover, if we estimate expectations concerning $X(t)$ at $N$ time points, the total query complexity scales on $N$ as $O(N^2)$, which is worse than the $O(N)$ complexity in the classical Monte Carlo method.
In this paper, to improve this, we propose a method to divide $U_{X(t)}$ based on orthogonal series density estimation.
This approach involves approximating the densities of $X(t)$ at $N$ time points with orthogonal series, where the coefficients are estimated as expectations of the orthogonal functions by QMCI.
By using these approximated densities, we can estimate expectations concerning $X(t)$ by QMCI without requiring deep circuits.
Our error and complexity analysis shows that to obtain the approximated densities at $N$ time points, our method achieves the circuit depth and total query complexity scaling as $O(\sqrt{N})$ and $O(N^{3/2})$, respectively.

\end{abstract}

%\keywords{}

\maketitle

\section{Introduction \label{sec:intro}}

Quantum Monte Carlo integration (QMCI) \cite{montanaro2015} built upon quantum amplitude estimation (QAE) \cite{brassard2002} is a prominent quantum algorithm to estimate expectations of random variables.
For an estimation with accuracy $\epsilon$, it runs with $O(1/\epsilon)$ query complexity, showing the quadratic speed-up compared to the classical counterpart with complexity $O(1/\epsilon^2)$.
One reason for its prominence is that it has the potential to provide quantum advantages in practical problems in science and industry.
For example, applications of QMCI to finance, especially financial derivative pricing, are widely investigated \cite{Rebentrost2018,Stamatopoulos2020optionpricingusing,Chakrabarti2021thresholdquantum,miyamoto2022bermudan,kaneko2022quantum,doriguello2022,wang2023option}.
In addition, its applications are considered in various fields including statistical physics~\cite{montanaro2015,Cornelissen2023}, lattice gauge theory~\cite{Gustafson2023}, nuclear physics~\cite{Du2024}, and machine learning~\cite{pmlr-v139-wang21w,wiedemann2023quantum,Wan_Zhang_Li_Zhang_Sun_2023,NEURIPS2023_401aa72e,hikima2024quantum}.

In many use cases of Monte Carlo integration, whether classical or quantum, we run the time evolution of a stochastic process $X(t)$.
For example, in derivative pricing, we evolve underlying asset prices because the derivative price is given as an expectation concerning the asset prices at a future time.
This paper focuses on such time evolution on a quantum computer.

%To explain the background, let us first see what time evolution in QMCI is.
%Denoting the stochastic process under consideration by $X(t)$, we suppose that we are given the discretized time points, the initial time $t_0$ and future times $t_1,\ldots,t_N$, and the transition probability density $p_i$ between $X(t_i)$ and $X(t_{i+1})$ for $i=0,\ldots,N-1$.
Let us suppose that we evolve $X(t)$ from $t=0$ to $t=T$ to calculate an expectation concerning $X(T)$.
We set the discretized time points $t_0=0,t_1,\ldots,t_N=T$ with small interval and evolve $X(t)$ with the transition probability densities $p_i$ between $X(t_i)$ and $X(t_{i+1})$.
To run QMCI, we construct the quantum circuits $U_{p_i}$ that encode $p_i$ into quantum states and combine them in sequence, yielding an oracle $U_{X(t_N)}$ to generate a quantum state $\ket{f_{X(t_N)}}$ that encodes the probability density $f_{X(t_N)}$ of $X(t_N)$ in the amplitudes.
$U_{X(t_N)}$ obviously has $O(N)$ circuit depth, and, when it is used in QMCI iteratively, the total depth of the quantum circuit becomes $O(N/\epsilon)$.
This depth growth in $N$, along with the dependency on $\epsilon$, may make it challenging to run long time evolutions in QMCI with high accuracy, especially in the early stages of fault-tolerant quantum computers (FTQCs).

There is another issue when we estimate multiple expectations related to $X(t_1),\ldots,X(t_N)$, the values of $X(t)$ at different times.
If we naively perform individual runs of QMCI, each of which uses $U_{X(t_i)}$ $O(1/\epsilon)$ times, the total number of queries to $\{U_{p_i}\}_i$ becomes $O(N^2/\epsilon)$.
This quadratic scaling on $N$ is worse than the classical Monte Carlo method: in the classical way, we store many realizations of $X(t_i)$ in memory and use them to calculate the expectation concerning $X(t_i)$ and also as initial values for time evolution after $t_i$, which leads to the total complexity scaling on $N$ as $O(N)$.
The worse scaling on $N$ in the quantum setting has already been pointed out in \cite{doriguello2022}, which considers expectation estimation at multiple times in optimal stopping problems.

Motivated by these issues, in this paper, we propose a way to divide the time evolution oracle $U_{X(t_N)}$ based on orthogonal series density estimation (OSDE) \cite{Efromovich2010}.
OSDE is a numerical technique to approximate unknown probability density functions with orthogonal function series\footnote{For another OSDE-based quantum algorithm for financial risk aggregation, see \cite{mori2024quantum}.}.
The idea is that at each time, we {\it save} the density $f_i$ of $X(t_i)$ as an orthogonal series. 
The coefficient of a basis function $P_l$ in the orthogonal series approximation of $f_i$ is given via the expectation of $P_l(X(t_i))$.
Therefore, we conceive the following procedure.
Using QMCI with $U_{p_0}$, we estimate the expectations of the basis functions and get an orthogonal series approximation $\hat{f}_1$ of $f_1$.
Then, by combining the quantum circuit $U^{\rm SP}_{\hat{f}_1}$ to generate a state encoding $\hat{f}_1$ with $U_{p_1}$, we get a quantum circuit $\tilde{U}_{X(t_2)}$, which approximately acts as $U_{X(t_2)}$, and by QMCI with this, we get an approximation $\hat{f}_2$ of $f_2$.
By repeating this, we get the approximated density functions $\hat{f}_1,\ldots,\hat{f}_N$ and the quantum circuits $U^{\rm SP}_{\hat{f}_1},\ldots,U^{\rm SP}_{\hat{f}_N}$, whose depth is not $O(N)$ but constant.
Using these, we can estimate expectations concerning $X(t_i)$.

Although this seems to work well, we are afraid that estimating the coefficients in $\hat{f}_i$ by QMCI may require deep quantum circuits and a large number of queries.
Note that the estimation of $\hat{f}_i$ depends on the previous estimation of $\hat{f}_{i-1}$ and the estimation error can accumulate.
Through the $N$-step estimation of $f_N$, the errors of $\epsilon^\prime$ in $N$ steps may accumulate to $O(N\epsilon^\prime)$.
If so, to suppress the overall error to a specified level $\epsilon$, we need to take $\epsilon^\prime\lesssim \epsilon/N$, which results in the circuit depth of order $O(N/\epsilon)$ and the total query number of order $O(N^2/\epsilon)$, the same as the naive way.

We evade this by utilizing unbiased QAE devised in \cite{Cornelissen2023}.
By this modified version of QAE, the bias of the error in estimation is suppressed while the variance of the error is kept similar to the original QAE.
By using QMCI based on unbiased QAE, the bias in $\hat{f}_i$ is suppressed and the variance accumulates linearly with respect to $N$, which means the typical error level of $\hat{f}_i$ scales on $N$ as $O(\sqrt{N})$.
Through a more detailed analysis of the error and complexity, we reveal that in our method, the circuit depth is of order $\widetilde{O}(\sqrt{N}/\epsilon)$ and the total query number is of order $\widetilde{O}(N^{3/2}/\epsilon)$, which show the improvement by a factor $\sqrt{N}$ compared to the naive way.

The rest of this paper is organized as follows.
Sec. \ref{sec:pre} is a preliminary one.
There, we explain Legendre polynomials, which we use as basis functions in OSDE in this paper, and the accuracy of function approximation by them, followed by a brief review of OSDE and quantum algorithms, (unbiased) QAE and QMCI.
In Sec. \ref{sec:QCTimeEv}, we explain quantum circuits for the time evolution of stochastic processes and issues in their implementation, which are briefly explained above, in more detail.
Then, in Sec. \ref{sec:main}, we elaborate our method to divide time evolution quantum circuits based on OSDE, presenting the concrete procedure as Algorithm \ref{alg:main} and the result of the rigorous analysis of the error and complexity as Theorem \ref{th:main}, whose proof is given in Appendix \ref{app:proofMain}.
We make a comparison between the proposed method and existing ones in Sec. \ref{sec:compare}.
In Sec.~\ref{sec:demo}, we conduct a numerical demonstration of our method, taking a kind of stochastic process used in some fields as an example, and compare it to other methods, seeing that our method becomes advantageous in some cases.
We summarize this paper in Sec. \ref{sec:sum}.

\section{Preliminary \label{sec:pre}}

\subsection{Notation}

$\mathbb{N}_0 := \{0\}\cup\mathbb{N}$ denotes the set of all non-negative integers.
For $n\in\mathbb{N}$, We define $[n]:=\{1,\cdots,n\}$ and $[n]_0:=\{0,1,\cdots,n\}$ for any $n\in \mathbb{N}_0$.
We also define %$\mathbb{N}_{\ge n}:=\{i\in \mathbb{N} \ | \ i\ge n\}$ for $n\in \mathbb{N}$. Similarly, we define
$\mathbb{R}_{>a}:=\{x\in\mathbb{R} \ | \ x>a\}$ and $\mathbb{R}_{\ge a}:=\{x\in\mathbb{R} \ | \ x\ge a\}$ for $a\in\mathbb{R}$.
$\mathbb{R}_+$ denotes the set of all positive real numbers, that is, $\mathbb{R}_{>0}$.

For $a,b\in \mathbb{N}_0$, $\delta_{a,b}$ denotes the Kronecker delta, which is 1 if $a=b$ or 0 otherwise.
For $d\in\mathbb{N}$ and $\vec{l}_1,\vec{l}_2\in \mathbb{N}_0^d$, we also define $\delta_{\vec{l}_1,\vec{l}_2}$, which is 1 if $\vec{l}_1=\vec{l}_2$ and 0 otherwise.

For $\vec{x}\in\mathbb{R}^d$, $\|\vec{x}\|$ denotes its Euclidean norm and $\|\vec{x}\|_\infty$ denotes its maximum norm.

$\vec{0}$ denotes the vector with all the entries equal to 0.

\subsection{Approximation of a function by Legendre expansion}

For $l\in\mathbb{N}_0$, the $l$-th Legendre polynomial $P_l$ is defined by
\begin{equation}
    P_l(x) := \frac{1}{2^l l!} \frac{d^l}{dx^l}\left(x^2-1\right)^l,
    \label{eq:LegPolyDef}
\end{equation}
where $x\in[-1,1]$.
The Legendre polynomials satisfy the following orthogonal relationship
\begin{equation}
    \int_{-1}^1 P_l(x)P_{l^\prime}(x) dx = \frac{2}{2l+1}\delta_{l,l^\prime}
\end{equation}
for any $l,l^\prime\in\mathbb{N}_0$.
In the multivariate case, where the number of the variables is $d\in\mathbb{N}$, we define the tensorized Legendre polynomials by
\begin{equation}
    P_{\vec{l}}(\vec{x}) := \prod_{i=1}^d P_{l_i}(x_i),
    \label{eq:LegPolyTenDef}
\end{equation}
where $\vec{l}=(l_1,\cdots,l_d)\in\mathbb{N}_0^d$ and $\vec{x}=(x_1,\cdots,x_d)\in \Omega_d := [-1,1]^d$.
The orthogonal relationship is now
\begin{equation}
    \int_{\Omega_d} P_{\vec{l}}(\vec{x})P_{\vec{l}^\prime}(\vec{x}) d\vec{x} = \frac{1}{C(\vec{l})}\delta_{\vec{l},\vec{l}^\prime},
    \label{eq:Orthddim}
\end{equation}
where $C(\vec{l}):=\prod_{i=1}^d \left(l_i+\frac{1}{2}\right)$.

Let us suppose that we are given a $d$-variate analytic function $f:\Omega_d\rightarrow\mathbb{R}$ and want to approximate it.
We now consider its Legendre expansion, that is, the series of the tensorized Legendre polynomials approximating $f$:
\begin{equation}
    f \approx \mathcal{P}_L[f] := \sum_{\vec{l}\in\Lambda_L} a_{f,\vec{l}} P_{\vec{l}},
\end{equation}
where we take the index set $\Lambda_L := [L]_0^d$ with $L\in\mathbb{N}$, and, for each $\vec{l}\in\Lambda_L$, the coefficient $a_{\vec{l}}$ is defined by
\begin{equation}
    a_{f,\vec{l}} := C(\vec{l}) \int_{\Omega_d} P_{\vec{l}}(\vec{x})f(\vec{x}) d\vec{x}.
    \label{eq:coef}
\end{equation}
The accuracy of the Legendre expansion was considered in \cite{wang2020analysis}, which in fact studied the expansion by more general orthogonal polynomials.
We now present the following theorem on the accuracy of the Legendre expansion, which is equivalent to Theorem 4.1 in \cite{wang2020analysis} restricted to the case of Legendre expansion.
The statement is informal, and the exact one is given in Appendix \ref{app:LegExAccu}.

\begin{theorem}[simplified]
For $f:\Omega_d\rightarrow\mathbb{R}$ with some properties, there exists $K\in\mathbb{R}_+$ and $\rho\in\mathbb{R}_{>1}$ such that, for any integer $L$ satisfying $L>\frac{d}{2\log \rho}$,
\begin{equation}
    \max_{\vec{x}\in\Omega_d} \left|\mathcal{P}_L[f](\vec{x})-f(\vec{x})\right| \le K\rho^{-L}
    \label{eq:LegExAccu}
\end{equation}
holds.
\label{th:LegExAccu}
\end{theorem}

Now, let us present some properties of the Legendre polynomials, which will be used later.
They are bounded as follows: for any $l\in\mathbb{N}_0$ and $x\in[-1,1]$,
\begin{equation}
    |P_l(x)| \le P_l(1) = 1,
\end{equation}
which is the specific case of Eq. 18.14.1 in \cite{NIST:DLMF}.
We thus have
\begin{equation}
    |P_{\vec{l}}(x)| \le 1
    \label{eq:PlvecBnd}
\end{equation}
for any $\vec{l}\in\mathbb{N}_0^d$.

Besides, for the later convenience, we define
\begin{equation}
        \Lambda^\prime_L := \Lambda_L\setminus\{\vec{0}\}.
    \end{equation}

\subsection{Orthogonal series density estimation \label{sec:OSDE}}

OSDE is a technique to approximate the probability density function $f_{\vec{X}}$ of a $\mathbb{R}^d$-valued random variable $\vec{X}$ by a series of orthogonal functions \cite{Efromovich2010}.
According to the range of the values $\vec{X}$ can take, we choose some orthogonal function system.
Hereafter, we assume that the range is a hyperrectangle and more specifically, it is $\Omega_d$, without loss of generality.
Then, as orthogonal functions, we take the tensorized Legendre polynomials: $f_{\vec{X}} \approx \mathcal{P}_L[f_{\vec{X}}]=\sum_{\vec{l} \in \Lambda_L} a_{f_{\vec{X}},\vec{l}} P_{\vec{l}}$ with $L$ set appropriately.
Note that each coefficient $a_{f_{\vec{X}},\vec{l}}$ is written as the expectation of the corresponding basis function $P_{\vec{l}}(\vec{X})$:
\begin{equation}
    a_{f_{\vec{X}},\vec{l}}=C(\vec{l}) \int_{\Omega_d} f_{\vec{X}}(\vec{x})P_{\vec{l}}(\vec{x}) d\vec{x} = C(\vec{l}) \times \mathbb{E}_{\vec{X}}\left[P_{\vec{l}}(\vec{X})\right],
    \label{eq:coefOSDE}
\end{equation}
where $\mathbb{E}_{\vec{X}}[\cdot]$ denotes the expectation with respect to the randomness of $\vec{X}$.

Although OSDE usually refers to obtaining an orthogonal series approximation of the density function with sampled values of $\vec{X}$, in this paper, we consider not the sample-based estimation but the quantum algorithm to estimate the coefficients $a_{f_{\vec{X}},\vec{l}}$.
While the error in sample-based OSDE includes the error due to the finite series approximation and the statistical error due to the finite sample size, it is now sufficient to consider only the former, which is bounded as Eq. \eqref{eq:LegExAccu} if $f_{\vec{X}}$ satisfies the condition in Theorem \ref{th:LegExAccu}.
The error due to estimating the coefficients by the quantum method is discussed later.

A common issue in OSDE is that the estimated density function $\hat{f}_{\vec{X}}$ might not be bona fide.
Namely, we are afraid that it might violate some conditions that any density function must satisfy naturally.
First, the positive definiteness
\begin{equation}
    \hat{f}_{\vec{X}}(\vec{x}) \ge 0, \forall \vec{x}\in\Omega_d,
    \label{eq:OSDEPosi}
\end{equation}
might not hold.
In the later discussion on the correctness of our quantum algorithm, we will make some assumptions on the true density function $f_{\vec{X}}$ so that Eq. \eqref{eq:OSDEPosi} holds when the algorithm guarantees the estimation accuracy to some extent.
On the other hand, the requirement that the total probability is 1, namely
\begin{equation}
    \int_{\Omega_d} \hat{f}_{\vec{X}}(\vec{x}) d\vec{x} = 1,
    \label{eq:OSDEProb1}
\end{equation}
can be satisfied easily when we use the Legendre polynomial system.
The lowest-order polynomial, $P_{\vec{l}}$ with $\vec{l}=\vec{0}$, is $P_{\vec{0}}=1$, and plugging this into Eq. \eqref{eq:coefOSDE} yields $a_{f_{\vec{X}},\vec{0}}=2^{-d}$ for any distribution.
Besides, $\int_{\Omega_d} P_{\vec{l}}(\vec{x}) d\vec{x}=0$ holds for any $\vec{l}\ne\vec{0}$.
Combining these, we see that $\hat{f}_{\vec{X}}$ written in the form of $\sum_{\vec{l}} a_{f_{\vec{X}},\vec{l}} P_{\vec{l}}$ with $a_{f_{\vec{X}},\vec{0}}=2^{-d}$ satisfies Eq. \eqref{eq:OSDEProb1}.

A widely used indicator of the accuracy of the estimated density function $\hat{f}_{\vec{X}}$ is the mean integrated squared error (MISE):
\begin{equation}
    \mathbb{E}_{\rm Q}\left[\int_{\Omega_d}\left(\hat{f}_{\vec{X}}(\vec{x})-f_{\vec{X}}(\vec{x})\right)^2 d\vec{x}\right].
\end{equation}
Here, $\mathbb{E}_{\rm Q}[\cdot]$ denotes the expectation with respect to the randomness in the algorithm.
In our quantum algorithm proposed later, the randomness stems from the building-block algorithm QAE, or more fundamentally, the quantum nature of the algorithm.
With the MISE bounded by $\epsilon^2$, the error of the expectation of a function $g(\vec{X})$ of $\vec{X}$ estimated with $\hat{f}_{\vec{X}}$ is also bounded:
\begin{align}
    &\left|\mathbb{E}_{\rm Q}\left[\int_{\Omega_d}\hat{f}_{\vec{X}}(\vec{x})g(\vec{x})d\vec{x}-\int_{\Omega_d}f_{\vec{X}}(\vec{x})g(\vec{x})d\vec{x} \right]\right| \nonumber \\
    \le & \left(\mathbb{E}_{\rm Q}\left[\int_{\Omega_d}\left(\hat{f}_{\vec{X}}(\vec{x})-f_{\vec{X}}(\vec{x})\right)^2 d\vec{x}\right]\right)^{\frac{1}{2}} \cdot \left(\int_{\Omega_d}\left(g(\vec{x})\right)^2 d\vec{x}\right)^{\frac{1}{2}} \nonumber \\
    =&O(\epsilon),
\end{align}
where we use the Cauchy–Schwarz inequality.
This matches some situations including financial derivative pricing described later, where we aim to get the derivative price as the expected payoff.

\subsection{Quantum circuits for arithmetics \label{sec:arithCircuit}}

In this paper, we consider computation on systems consisting of multiple quantum registers. 
We use the fixed-point binary representation for real numbers and, for each $x\in\mathbb{R}$, we denote by $\ket{x}$ the computational basis state on a quantum register that holds the bit string equal to the binary representation of $x$.
For $\vec{x}=(x_1,\cdots,x_d)\in\mathbb{R}^d$, we define the quantum state on a $n$-register system $\ket{\vec{x}}:=\ket{x_1}\cdots\ket{x_d}$.
We assume that every register has a sufficiently large number of qubits and thus neglect errors caused by finite-precision representation.

We can perform arithmetic operations on numbers represented on registers.
For example, we can implement quantum circuits for four basic arithmetic operations such as addition $O_{\mathrm{add}}:\ket{a}\ket{b}\ket{0}\mapsto\ket{a}\ket{b}\ket{a+b}$ and multiplication $O_{\mathrm{mul}}:\ket{a}\ket{b}\ket{0}\mapsto\ket{a}\ket{b}\ket{ab}$, where $a,b\in\mathbb{Z}$.
For concrete implementations, see \cite{MunosCoreas2022} and the references therein.
In the finite-precision binary representation, these operations are immediately extended to those for real numbers.
Furthermore, using the above circuits, we obtain a quantum circuit $U_p$ to compute a polynomial $p(x)=\sum_{n=0}^N a_n x^n$ on real numbers $x$: $U_p\ket{x}\ket{0}=\ket{x}\ket{p(x)}$.
Also for $p$ as an elementary function such as $\exp$, $\sin$/$\cos$, and so on, we have a similar circuit to compute it, given a piecewise polynomial approximation of $p$ \cite{haner2018optimizing}.
Hereafter, we collectively call these circuits arithmetic circuits.

\subsection{Quantum amplitude estimation with suppressed bias}

QAE \cite{brassard2002} is a quantum algorithm to estimate the squared amplitude of a target basis state in a quantum state and is used in various quantum algorithms as a core subroutine.
In this paper, to be concrete, we consider the case where the target basis state and the other ones are distinguished by whether a specific qubit takes $\ket{1}$ or $\ket{0}$.
That is, we suppose that we are given the following oracle $A$ to generate a quantum state $\ket{\psi}$ on a system $S$ consisting of a register $R_1$ and a qubit $R_2$:
\begin{equation}
    A \ket{0}\ket{0} = \sqrt{a}\ket{\psi_1}\ket{1} + \sqrt{1-a} \ket{\psi_0}\ket{0} =: \ket{\psi},
    \label{eq:StatePrepOra}
\end{equation}
where $\ket{\psi_0}$ and $\ket{\psi_1}$ are some quantum states on $R_1$.
Although QAE can be considered in more general settings, the above setting is enough for the purpose of this paper.
Then, we can construct a quantum algorithm that outputs an estimation $\hat{a}$ of $a$ within accuracy $\epsilon$ with high probability, making $O(1/\epsilon)$ uses of $A$ along with some elementary quantum gates.

The original version of QAE only guarantees that $\hat{a}$ is close to $a$ within distance $\epsilon$ with high probability and does not care about the bias of $\hat{a}$; $\mathbb{E}[\hat{a}-a]$ might not be 0.
Some previous works proposed versions of QAE whose outputs have reduced biases~\cite{callison2022improved,Cornelissen2023,lu2023random}.
In particular, Ref. \cite{Cornelissen2023} presented a version of QAE that outputs an estimation with the bias suppressed to an arbitrary level with the logarithmic overhead in the query complexity.
We here present a theorem on its accuracy and complexity taken from Ref. \cite{Cornelissen2023} with some modification.

\begin{theorem}[Theorem 2.2 in \cite{Cornelissen2023}, modified]

Let $\delta,\epsilon\in(0,1)$.
Suppose that we have access to the oracle $A$ that acts as Eq. \eqref{eq:StatePrepOra} with $a\in(0,1)$.
Then, there exists a quantum algorithm $\proc{UBQAE}(A,\epsilon,\delta)$ that outputs a random real number $\hat{a}$ such that
\begin{equation}
    \left|\mathbb{E}[\hat{a}-a]\right| \le \delta
    \label{eq:QAEBias}
\end{equation}
and
\begin{equation}
    \mathbb{E}\left[\left(\hat{a}-a\right)^2\right] \le \epsilon^2 + \delta
    \label{eq:QAEVar}
\end{equation}
holds, making
\begin{equation}
    O\left(\frac{1}{\epsilon}\log\log\left(\frac{1}{\epsilon}\right)\log\left(\frac{1}{\delta\epsilon}\right)\right)
\end{equation}
queries to $A$.
\label{th:UBQAE}
\end{theorem}

We give some comments on the difference between Theorem \ref{th:UBQAE} in the above and Theorem 2.2 in \cite{Cornelissen2023}.

First, in Theorem 2.2 in \cite{Cornelissen2023}, it is assumed that we have access to the two reflection operators, which do not appear in Theorem \ref{th:UBQAE} apparently.
The first one is $I-2\ket{\psi}\bra{\psi}$, the reflection operator with respect to $\ket{\psi}$, and, in the current case, constructed with $O(1)$ uses of $A$ and some elementary gates \cite{brassard2002}.
The second one is the reflection operator with respect to the target state.
In the current case, it is $U_1:=I_S-2I_{R_1}\otimes \ket{1}\bra{1}$, where $I_S$ and $I_{R_1}$ are the identity operators on $S$ and $R_1$, respectively, and $U_1$ is just the Pauli $Z$ gate on $R_2$.
Thus, assuming naturally that $A$ is much more costly than $U_1$, we do not consider the number of queries to it in Theorem \ref{th:UBQAE}.

Second, the method considered in Theorem 2.2 in \cite{Cornelissen2023} is non-destructive: it is assumed that a single copy of the quantum state $\ket{\psi}$ is given initially, and after the estimation of $a$, $\ket{\psi}$ is restored.
In this paper, we do not require this property, and we use the oracle $A$ to generate $\ket{\psi}$ repeatedly.
Thus, we do not perform the process to restore $\ket{\psi}$ in the non-destructive coin flip, a subroutine in the method in \cite{Cornelissen2023}, which is an iterative procedure with iteration number unbounded.
Therefore, we give just a deterministic upper bound on the query complexity in Theorem \ref{th:UBQAE}, while Theorem 2.2 in \cite{Cornelissen2023} gives the complexity upper bound as an expectation.

\subsection{Quantum Monte Carlo Integration \label{sec:QMCI}}

Among quantum algorithms built upon QAE, quantum Monte Carlo integration (QMCI) \cite{montanaro2015} is particularly prominent.
We now briefly outline it.
The aim of QMCI is to calculate the expected value of a random variable.
Let us suppose that we want to calculate the expectation $\mathbb{E}_{\vec{X}}[g(\vec{X})]$ for $g:\Omega_d\rightarrow[0,1]$.
To use QMCI, we assume that we are given the following quantum circuits.
The first one is the state preparation oracle $U_{\vec{X}}$ that acts as
\begin{equation}
    U_{\vec{X}} \ket{0} = \sum_{\vec{x} \in \mathcal{X}} \sqrt{\tilde{f}_{\vec{X}}(\vec{x})} \ket{\vec{x}}.
    \label{eq:SPOra}
\end{equation}
Here, $\mathcal{X} \subset \Omega_d$ is a finite set and $\tilde{f}_{\vec{X}}$ is a map from $\mathcal{X}$ to $[0,1]$ such that $\sum_{\vec{x}\in\mathcal{X}}\tilde{f}_{\vec{X}}(\vec{x})=1$.
While $\vec{X}$ is originally a random variable that takes continuous values in $\Omega_d$ with density $f_{\vec{X}}$, we here introduce a kind of discretized approximation such that $\vec{X}$ takes only values $\vec{x}$ in $\mathcal{X}$ with probability $\tilde{f}_{\vec{X}}(\vec{x})$. 
The second oracle we use is $U_g$ that acts as
\begin{equation}
    U_g \ket{\vec{x}} \ket{0} = \ket{\vec{x}} \left(\sqrt{g(\vec{x})}\ket{1}+\sqrt{1-g(\vec{x})}\ket{0}\right)
    \label{eq:IntegrandOra}
\end{equation}
for any $\vec{x}\in\mathcal{S}$.
Combining $U_{\vec{X}}$ and $U_g$, we construct a quantum circuit to generate the quantum state
\begin{equation}
    \sum_{\vec{x} \in \mathcal{S}} \sqrt{\tilde{f}_{\vec{X}}(\vec{x}) g(\vec{x})} \ket{\vec{x}}\ket{1}+\sum_{\vec{x} \in \mathcal{S}}\sqrt{\tilde{f}_{\vec{X}}(\vec{x}) \left(1-g(\vec{x})\right)} \ket{\vec{x}}\ket{0},
\end{equation}
and using it, we estimate $\sum_{\vec{x} \in \mathcal{S}}\tilde{f}_{\vec{X}}(\vec{x}) g(\vec{x})$ as an approximation of $\mathbb{E}_{\vec{X}}[g(\vec{X})]$ by QAE.

These oracles are in fact implementable in some cases.
For $U_{\vec{X}}$, many methods for function encoding into quantum states have been proposed so far, including the pioneering Grover-Rudolph method \cite{grover2002creating} and other improved methods~\cite{Sanders2019,zoufal2019quantum,Holmes2020,Endo2020,Rattew2021efficient,kaneko2022quantum,rattew2022preparing,mcardle2022quantum,MarinSanchez2023,Moosa_2024}.
In these methods, the gate cost scales logarithmically with respect to the number of the points in $\mathcal{X}$, which means that we can take an exponentially fine grid as $\mathcal{X}$ and thus $\sum_{\vec{x} \in \mathcal{S}}\tilde{f}_{\vec{X}}(\vec{x}) g(\vec{x})$ approximates $\mathbb{E}_{\vec{X}}[g(\vec{X})]$ well.

Let us make some additional remarks on $U_{\vec{X}}$.
First, more generally than Eq.~\eqref{eq:SPOra}, we may consider $U_{\vec{X}}$ that, accompanying ancilla qubits, acts as
\begin{align}
U_{\vec{X}} \ket{0}\ket{0} = \sum_{\vec{x} \in \mathcal{X}} \sqrt{\tilde{f}_{\vec{X}}(\vec{x})} \ket{\vec{x}}\ket{\phi_{\vec{x}}},
\label{eq:SPOra2}
\end{align}
where $\ket{\phi_{\vec{x}}}$ is a quantum state on the ancillas that may depend on $\vec{x}$. This makes no essential change in the following discussion, and we assume the form of $U_{\vec{X}}$ as Eq.~\eqref{eq:SPOra} for simplicity.
Second, the ability to efficiently prepare such a state may seem to contradict the previous work \cite{Herbert2021NoQuantumSpeedup}, which argues that the Grover-Rudolph method \cite{grover2002creating} involved in QMCI does not yield any quantum speed-up. In the Grover-Rudolph method, we compute the integrals of the density $f_{\vec{X}}$ over small intervals, for which Ref.~\cite{grover2002creating} originally proposed to use the classical Monte Carlo method, and, according to Ref.~\cite{Herbert2021NoQuantumSpeedup}, this ruins the quantum speed-up. Fortunately, there are some ways to avoid this. Ref.~\cite{kaneko2022quantum} pointed out that, for some distributions, the aforementioned integrals of $f_{\vec{X}}$ can be analytically computed, which enables efficient encoding of $f_{\vec{X}}$. The methods in \cite{Holmes2020,mcardle2022quantum,Moosa_2024} approximate $f_{\vec{X}}$ with functions that can be easily encoded. Also note that, if $\vec{X}$ is computed by an analytic formula $\vec{X}=h(\vec{W})$ from another random variable $\vec{W}$ for which we have the encoding oracle $U_{\vec{W}}$, we also have the oracle like Eq.~\eqref{eq:SPOra2} by combining $U_{\vec{W}}$ with the arithmetic circuit for $h$.

When it comes to $U_g$, if we assume that $g$ is explicitly given as an elementary function, $U_g$ is constructed with arithmetic circuits and $U_{\rm rot}$.
Here, $U_{\rm rot}$ is the Y-rotation gate with controlled angle, which acts as $U_{\rm rot}\ket{\theta}\ket{\psi}=\ket{\theta}\begin{pmatrix} \cos(\theta/2) & -\sin(\theta/2) \\  \sin(\theta/2) & \cos(\theta/2)\end{pmatrix} \ket{\psi}$ for any $\theta\in\mathbb{R}$ and 1-qubit state $\ket{\psi}$, and its implementation is considered in \cite{woerner2019quantum}.

\subsection{Applying QMCI to derivative pricing \label{sec:QMCI4Pricing}}

Among various applications of QMCI proposed so far, applications to derivative pricing are particularly prominent \cite{Rebentrost2018,Stamatopoulos2020optionpricingusing,Chakrabarti2021thresholdquantum,miyamoto2022bermudan,kaneko2022quantum,doriguello2022,wang2023option}.
Although we leave a detailed explanation of it to textbooks such as \cite{hull2003options}, we now provide the outline briefly.
A derivative is a two-party contract in which a party receives payoffs determined by the prices of some underlying assets (say, stocks) from the other.
A simple example is a European call option: it grants a party the right to buy an asset at a prefixed price $K$ at a future time $T$, which is equivalent to getting a payoff $\max\{X(T)-K,0\}$, where $X(t)$ is the asset price at time $t$.
To diversify trading strategies and improve risk management, various kinds of derivatives are traded in the financial market.
To price a derivative with a payoff $g_{\rm pay}(\vec{X}(T))$ determined by $\vec{X}(T)=(X_1(T),\cdots,X_d(T))$, the prices of the $d$ underlying assets at $T$, the standard approach is as follows.
We first model the dynamics of $\vec{X}(t)$ with a stochastic differential equation (SDE)
\begin{align}
d\vec{X}(t)=\vec{\mu}(t,\vec{X}(t))dt+\Sigma(t,\vec{X}(t))d\vec{W}(t),
\label{eq:SDE}
\end{align}
where $\vec{W}(t)$ is a $d^\prime$-dimensional Brownian motion, and $\vec{\mu}$ and $\Sigma$ are $\mathbb{R}^d$ and $\mathbb{R}^{d \times d^\prime}$-valued functions of $(t, \vec{X}(t))$, respectively.
Then, the derivative price is given by the expected payoff\footnote{Strictly speaking, the discount factor must be involved.} $\mathbb{E}_{\vec{X}(T)}[g_{\rm pay}(\vec{X}(T))]$.

We can estimate this type of expectation by QMCI if we have quantum circuits $U_{\vec{X}(T)}$ like Eq.~\eqref{eq:SPOra} and $U_{g_{\rm pay}}$ like Eq.~\eqref{eq:IntegrandOra}.
We can in fact implement $U_{\vec{X}(T)}$ for various models of the asset price dynamics as discussed in \cite{Rebentrost2018,Stamatopoulos2020optionpricingusing,Chakrabarti2021thresholdquantum,kaneko2022quantum,wang2023option}.
We discretize Eq.~\eqref{eq:SDE} in time by some scheme, e.g., the Euler-Maruyama scheme~\cite{Maruyama1955} $\vec{X}(t+\Delta t)=\vec{X}(t)+\vec{\mu}(t,\vec{X}(t))\Delta t+\Sigma(t,\vec{X}(t))\Delta \vec{W}$, where $\Delta \vec{W}$ is a vector of normal random variables, and iterating this discretized evolution yields $\vec{X}(T)$.
Since the normal distribution can be efficiently encoded into a quantum state~\cite{kaneko2022quantum,Holmes2020,Rattew2021efficient,MarinSanchez2023}, we can construct a quantum circuit to encode the density of $\vec{X}(T)$, namely $U_{\vec{X}(T)}$, similarly to Eq.~\eqref{eq:SPOra2}.
Becides, because $g_{\rm pay}$ is usually given as a simple function, $U_{g_{\rm pay}}$ can be implemented with arithmetic circuits mentioned in Sec.~\ref{sec:arithCircuit}.

\section{Quantum circuit for the time evolution of a stochastic process \label{sec:QCTimeEv}}

We now focus on how to implement the state preparation oracle $U_{\vec{X}}$ in the case that $\vec{X}$ is a value of some stochastic process $\vec{X}(t)$ at a fixed time $t=T$, as in derivative pricing mentioned above.
This is motivated by the issue described below.

In this case, usually, we implement $U_{\vec{X}}$ as a circuit to simulate the time evolution of the stochastic process.
The meaning of this is as follows.
Let us consider a $\Omega_d$-valued stochastic process $\{\vec{X}(t_i)\}_{i\in[N]_0}$ on the $(N+1)$-point discrete time $t_0,\cdots,t_N=T$, with the initial value $\vec{X}(t_0)$ being deterministic.
We denote by $p_i(\cdot|\cdot)$ the conditional transition probability density at each time $t_i$:
\begin{equation}
    {\rm Pr}\left\{\vec{X}(t_{i+1}) \in \mathcal{A} \middle| \vec{X}(t_i)=\vec{x}_i\right\} = \int_\mathcal{A} p_i(\vec{x}_{i+1}|\vec{x}_i) d \vec{x}_{i+1},
    \label{eq:TranProb}
\end{equation}
where $\mathcal{A}\subset\Omega_d$.
The density function of $\vec{X}(t_i)$ is then given by
\begin{align}
    &f_i(\vec{x}):= \nonumber \\
    &
    \begin{dcases}
        p_0(\vec{x}_1|\vec{x}_0) & ; \ i=1 \\
        \int_{\Omega_d} d\vec{x}_1\cdots \int_{\Omega_d} d\vec{x}_{i-1} \prod_{j=0}^{i-1} p_j(\vec{x}_{j+1}|\vec{x}_j) & ; \ i=2,...,N
    \end{dcases}
    .
\end{align}
We assume that, as is the case when $p_i$ is given via some SDE as described in Sec.~\ref{sec:QMCI4Pricing}, we can construct the oracle $U_{p_i}$ to generate the quantum state encoding the transition probability:
\begin{equation}
    U_{p_i}\ket{\vec{x}_i}\ket{0}=\sum_{\vec{x}_{i+1}\in\mathcal{X}_{i+1}} \sqrt{\tilde{p}_{i}(\vec{x}_{i+1}|\vec{x}_i)} \ket{\vec{x}_i}\ket{\vec{x}_{i+1}},
    \label{eq:Upi}
\end{equation}
where, again, we introduce the discretized approximation such that conditioned on $\vec{X}(t_{i})=\vec{x}_i$, $\vec{X}(t_{i+1})$ takes the values in a finite set $\mathcal{X}_{i+1}$ with a probability mass function $\tilde{p}_{i+1}(\cdot|\vec{x}_i)$. 
Given $U_{p_0},\cdots,U_{p_{N-1}}$, we combine them into the quantum circuit shown in Fig. \ref{fig:UvecXWhole} to obtain the state preparation oracle $U_{\vec{X}}$ for $\vec{X}=\vec{X}(t_N)$: 
\begin{align}
    & U_{\vec{X}(t_N)}\ket{\vec{x}_0}\ket{0}^{\otimes N} \nonumber \\
    =& \sum_{\vec{x}_1\in\mathcal{X}_1} \cdots \sum_{\vec{x}_{N}\in\mathcal{X}_{N}}  \sqrt{\prod_{i=0}^{N-1}\tilde{p}_i(\vec{x}_{i+1}|\vec{x}_i)}  \ket{\vec{x}_0}\ket{\vec{x}_1}\cdots\ket{\vec{x}_{N}} \nonumber \\
    =&\sum_{\vec{x}_{N}\in\mathcal{X}_{N}} \sqrt{\tilde{f}_{N}(\vec{x}_{N})} \ket{\psi(\vec{x}_{N})} \ket{\vec{x}_{N}}=:\ket{\tilde{f}_N},
    \label{eq:UXFromUpi}
\end{align}
where for $j\in[N]$,
\begin{equation}
    \tilde{f}_{j}(\vec{x}_{j}) := \sum_{\vec{x}_1\in\mathcal{X}_1} \cdots \sum_{\vec{x}_{j-1}\in\mathcal{X}_{j-1}} \prod_{i=0}^{j-1}\tilde{p}_i(\vec{x}_{i+1}|\vec{x}_i)
\end{equation}
and
\begin{align}
    &\ket{\psi(\vec{x}_{j})} := \frac{1}{\sqrt{\tilde{f}_j(\vec{x}_{j})}} 
\times \nonumber \\
    & \sum_{\vec{x}_1\in\mathcal{X}_1} \cdots \sum_{\vec{x}_{j-1}\in\mathcal{X}_{j-1}}  \sqrt{\prod_{i=0}^{j-1}\tilde{p}_i(\vec{x}_{i+1}|\vec{x}_i)}  \ket{\vec{x}_0}\ket{\vec{x}_1}\cdots\ket{\vec{x}_{j-1}}.
\end{align}
Regarding the first register in the last line in Eq. \eqref{eq:UXFromUpi} as ancillary and adopting the discrete approximation that $\vec{X}(t_N)$ takes $\vec{x}_{N}\in\mathcal{X}_{N}$ with probability $\tilde{f}_{N}(\vec{x}_{N})$, we can see $U_{\vec{X}(t_N)}$ in Eq. \eqref{eq:UXFromUpi} as a state-preparation oracle like Eq. \eqref{eq:SPOra}.

\begin{figure}
\captionsetup{justification=raggedright}
\begin{center}
\subfigure[]{
\begin{quantikz}[row sep={0.5cm,between origins}, column sep={0.3cm}]
\lstick{$\ket{\vec{x}_i}$} & \gate[2]{U_{p_i}} & \qw & \rstick{$\ket{\vec{x}_i}$}\\
\lstick{$\ket{0}$} & & \qw & \rstick{$\sum_{\vec{x}_{i+1}\in\mathcal{X}_{i+1}} \sqrt{\tilde{p}(\vec{x}_{i+1}|\vec{x}_{i})} \ket{\vec{x}_{i+1}}$}
\end{quantikz}
}
\subfigure[]
{
\begin{quantikz}[row sep={0.5cm,between origins}, column sep={0.3cm}]
\lstick{$\ket{\vec{x}_0}$} & \gate[2]{U_{p_0}} & \qw  & \qw & \cdots & & \qw & \qw & \rstick{$\vec{x}_0$}\\
& & \gate[2]{U_{p_1}} & \qw  & \cdots & & \qw & \qw & \rstick{$\vec{x}_1$}\\
& \qw & \qw & \qw & \cdots & & \qw & \qw & \rstick{$\vec{x}_{2}$}\\
& \vdots & & & & & & \\
& \qw & \qw & \qw & \cdots & & \gate[2]{U_{p_{N-1}}} & \qw & \rstick{$\vec{x}_{N-1}$}\\
& \qw & \qw & \qw & \cdots & & & \qw& \rstick{$\vec{x}_{N}$}
\end{quantikz}
\label{fig:UvecX}
}
\caption{(a) The quantum circuit for time evolution over a time step, which, given the value of $\vec{X}(t_{i})$, generates the superposition of the possible values of $\vec{X}(t_{i+1})$ like Eq. \eqref{eq:Upi}. (b) The state preparation circuit $U_{\vec{X}(t_N)}$ for $\vec{X}(t_{i+1})$  is constructed by combining $U_{p_0},\cdots,U_{p_{N-1}}$ in series.}
\label{fig:UvecXWhole}
\end{center}
\end{figure}

Hereafter, for a set $\mathcal{U}$ of quantum circuits, we say that the $\mathcal{U}$-depth of a quantum circuit $C$ is $D$ if $D$ elements in $\mathcal{U}$ are combined in series in $C$.
For example, the $\mathcal{U}_p$-depth of the circuit in Fig. \ref{fig:UvecX} is $N$, where $\mathcal{U}_p:=\{U_{p_0},\ldots,U_{p_{N-1}}\}$.
If $\mathcal{U}$ consists of one circuit $U$, we use the term $U$-depth.

The quantum circuit $U_{\vec{X}(t_N)}$ has $O(N)$ $\mathcal{U}_p$-depth, and for large $N$, this might be an issue in implementing $U_{\vec{X}(t_N)}$ on real quantum hardware, especially in the early-FTQC era, in which the circuit depth might be limited.
By this, we are motivated to divide $U_{\vec{X}(t_N)}$ into some pieces.

In addition to this hardware perspective, we point out the following issue: in quantum computing, we need to simulate $\vec{X}(t)$ from the initial time $t_0$ every time we get the distributions at $t_1,\ldots,t_N$.
Although in the above, we have described the state preparation for $\vec{X}(t)$ at the terminal time $t=t_N$, we often need the states encoding the distributions of $\vec{X}(t)$ at the intermediate times $t_1,\ldots,t_{N-1}$ too.
For example, in practice, we often price derivative contracts that are written on the same asset but have different maturities.
That is, for derivatives with payoffs $g_1(\vec{X}(t_1)),\ldots,g_N(\vec{X}(t_N))$ determined by the values of the asset price $\vec{X}(t)$ at different times $t=t_1,\ldots,t_N$, we estimate their prices $V_1:=\mathbb{E}_{\vec{X}(t_1)}[g_1(\vec{X}(t_1))],\ldots,V_N:=\mathbb{E}_{\vec{X}(t_N)}[g_N(\vec{X}(t_N))]$.
In runs of QMCI to estimate these, we use $U_{\vec{X}(t_1)},\ldots,U_{\vec{X}(t_N)}$ to generate the quantum states $\ket{f_1},\ldots,\ket{f_N}$ encoding the distributions at $t_1,\ldots,t_N$.
Note that $U_{\vec{X}(t_1)},\ldots,U_{\vec{X}(t_N)}$ are operated on the same initial state $\ket{\vec{x}_0}\ket{0}^{\otimes N}$.
They have $O(N)$ $\mathcal{U}_p$-depth each and make $O(N^2)$ queries to oracles in $\mathcal{U}_p$ in total.
This scaling on $N$ is worse than that of the sample complexity of classical Monte Carlo integration, which scales as $O(N)$.
Classically, we can store the intermediate values of $\vec{X}(t)$ on memory: having $\vec{X}^{(1)}(t_i),\ldots,\vec{X}^{(N_{\rm path})}(t_i)$, $N_{\rm path}$ sample values of $\vec{X}(t_i)$, we can use them to calculate $\frac{1}{N_{\rm path}}\sum_{j=1}^{N_{\rm path}} g_i(\vec{X}^{(j)}(t_i))$ as an approximation of $V_i$ and also get $\vec{X}^{(1)}(t_{i+1}),\ldots,\vec{X}^{(N_{\rm path})}(t_{i+1})$, sample values at the next time $t_{i+1}$, evolved from $\vec{X}^{(1)}(t_i),\ldots,\vec{X}^{(N_{\rm path})}(t_i)$.
This means that in the classical Monte Carlo method for $V_1,\ldots,V_N$, the total number of sampling scales on $N$ as $O(N)$.

\section{Dividing quantum circuits for the time evolution of stochastic processes \label{sec:main}}

\subsection{Idea}

Motivated by the above issues, we propose a method to divide $U_{\vec{X}(t_N)}$ based on OSDE\footnote{We refer to a previous work \cite{Herbert2022quantummontecarlo}, which also proposed using function approximation with a series of basis functions in the context of QMCI. In their method, the series approximation of not the density but the integrand is classically precomputed. Then, the expectations of basis functions are estimated by QMCI, and these estimates are combined with series expansion coefficients, yielding an estimate of the wanted integral.}.

Our aim is to generate the quantum state $\ket{f_N}$ that encodes the probability density $f_N$ of the value of the stochastic process $\vec{X}(t)$ at time $t=t_N$.
Since we are given only its initial value $\vec{X}(0)$ and the transition probability $p_i$, we need to apply the $N$ operators $U_{p_0},\ldots,U_{p_{N-1}}$.
However, what if we know the probability density $f_i$ of $\vec{X}(t)$ at some intermediate time $t_i,i\in[N-1]$?
If so, starting from the quantum state $\ket{f_i}$ encoding $f_i$, we can get $\ket{f_N}$ only by applying the operators $U_{p_i},\ldots,U_{p_{N-1}}$ , which results in a smaller circuit depth and query number.
Of course, we do not know any of the intermediate density functions $\{f_i\}_{i\in[N-1]}$ beforehand.
Then, how about estimating them?
That is, we conceive the method described as the loop of the following steps:
\begin{enumerate}
    \item Given (an approximation of) $f_i$, generate the quantum state $\ket{f_i}$ that encodes it.
    \item By operating $U_{p_i}$, generate the quantum state $\ket{f_{i+1}}$ that encodes $f_{i+1}$
    \item Using $\ket{f_{i+1}}$, get an approximation of $f_{i+1}$.
\end{enumerate}
We require that, in step 3, we get some classical data that determine the approximating function.
This means that each round of this loop starts by inputting the classical data and ends by outputting the other classical data.
If this is possible, we can realize the rounds of the above loop as separate quantum algorithms, which means that dividing $U_{\vec{X}(t_N)}$ is achieved.

We then concretize the above rough sketch of the method.
In particular, given the quantum state $\ket{f_{i+1}}$, how can we get an approximation function of $f_{i+1}$ in a way that allows for some classical description?
We can do this by OSDE described in Sec. \ref{sec:OSDE}.
In the quantum algorithm presented later, we estimate the Legendre expansion coefficients $\{a_{f_{i+1},\vec{l}}\}_{\vec{l}}$ for $f_{i+1}$, which is given via the expectations of $\{P_{\vec{l}}(\vec{X}(t_{i+1}))\}_{\vec{l}}$ as shown in \eqref{eq:coefOSDE}, using QMCI.
These coefficients are just a set of real numbers, and thus we can have them as classical data.

\subsection{Quantum algorithm}

Now, we present a detailed description of the aforementioned method.
We start by presenting the informal version of our main theorem on the accuracy and complexity of our algorithm.
In the proof, we only show the procedure of our algorithm.
The exact statement of the theorem and the rest of the proof, which is on the accuracy and complexity, are given in Appendix \ref{app:proofMain}.

\begin{theorem}[simplified]
    Let $\alpha,\delta,\epsilon\in(0,1)$.
    Let $\{\vec{X}(t_{i})\}_{i\in[N]_0}$ be a $\Omega_d$-valued stochastic process with the deterministic initial value $\vec{x}_0$ and the conditional transition probability density $p_i(\cdot|\cdot)$.
    Set $L\in\mathbb{N}$ sufficiently large.
    Suppose that the following assumptions on the access to oracles hold:
    \begin{itemize}
        \item 
        For each $i\in[N-1]_0$, we have access to the oracle $U_{p_i}$ acting as Eq. \eqref{eq:Upi}.
        Here, $\mathcal{X}_1,\ldots,\mathcal{X}_N\subset\Omega_d$ are finite sets, $\mathcal{X}_0:=\{\vec{x}_0\}$, and $\tilde{p}_i:\mathcal{X}_{i+1}\times\mathcal{X}_i\rightarrow[0,1],i\in[N-1]_0$ satisfies $\sum_{\vec{x}_{i+1}\in\mathcal{X}_{i+1}}\tilde{p}_i(\vec{x}_{i+1}|\vec{x}_i)=1$ for any $\vec{x}_i\in\mathcal{X}_i$ and approximates $p_i$.

        \item For $f:\Omega_d\rightarrow\mathbb{R}_{\ge 0}$ written as $f=\sum_{\vec{l}\in\Lambda_L} a_{\vec{l}} P_{\vec{l}}$ with any real number set $\{a_{\vec{l}}\}_{\vec{l}\in\Lambda_L}$ including $a_{\vec{0}}=2^{-d}$, we have access to the oracle $U^{\rm SP}_{f}$ that acts as
        \begin{equation}
            U^{\rm SP}_{f}\ket{0} = \ket{f} := \sum_{\vec{x}\in\mathcal{X}_{f}} \sqrt{\tilde{f}(\vec{x})} \ket{\vec{x}}.
            \label{eq:USPf}
        \end{equation}
        Here, $\mathcal{X}_{f} \subset \Omega_d$ is a finite set, and $\tilde{f}:\mathcal{X}_{f}\rightarrow[0,1]$ is a map that satisfies $\sum_{\vec{x}\in\mathcal{X}_{f}}\tilde{f}(\vec{x})=1$ and approximates $f$.

        \item For any $\vec{l}\in\mathbb{N}_0^d$, we have access to the oracle $U^{\rm amp}_{P_{\vec{l}}}$ that acts as
        \begin{equation}
        U^{\rm amp}_{P_{\vec{l}}}\ket{\vec{x}}\ket{0} =  \ket{\vec{x}} \left(\sqrt{\frac{1+P_{\vec{l}}(\vec{x})}{2}}\ket{1}+\sqrt{\frac{1-P_{\vec{l}}(\vec{x})}{2}}\ket{0}\right)
        \label{eq:UPl}
        \end{equation}
        for any $\vec{x}\in\mathbb{R}^d$.
        
    \end{itemize}
    Then, there exists a quantum algorithm that, with probability at least $1-\alpha$, outputs approximation functions $\hat{f}_1,\ldots,\hat{f}_N:\Omega_d\rightarrow\mathbb{R}_{\ge0}$ for $f_1,\ldots,f_N$ with the following properties: each $\hat{f}_i$ is written as $\hat{f}_i=\sum_{\vec{l}\in\Lambda_L} \hat{a}_{i,\vec{l}}P_{\vec{l}}$ with random coefficients $\{\hat{a}_{i,\vec{l}}\}_{\vec{l}\in\Lambda^\prime_L}$ and $\hat{a}_{i,\vec{0}}=2^{-d}$, and its MISE is bounded as
    \begin{equation}
        \mathbb{E}_{\rm Q}\left[\int_{\Omega_d} \left(\hat{f}_i(\vec{x})-f_i(\vec{x})\right)^2 d\vec{x}\right]\le\epsilon^2,
        \label{eq:MISE}
    \end{equation}
    where $\mathbb{E}_{\rm Q}[\cdot]$ denotes the expectation with respect to the randomness in the algorithm.
    In this algorithm, we use quantum circuits with $\mathcal{U}_{\rm all}$-depth of order 
    \begin{align}
        O&\left(\frac{\sqrt{N}}{\epsilon}L^d\log^{\frac{d}{2}}L \times \right. \nonumber \\
        &\quad\left.\log\log\left(\frac{\sqrt{N}}{\epsilon}L^d\log^{\frac{d}{2}}L\right) \log\left(\frac{N}{\epsilon}L^d\log^{\frac{d}{2}}L\right)\right)
        \label{eq:circuitDepth},
    \end{align}
    and the total number of queries to oracles in $\mathcal{U}_{\rm all}$ is
    \begin{align}
        O&\left(\frac{N^{3/2}}{\epsilon}L^{2d}\log^{\frac{d}{2}}L \times \right. \nonumber \\
        &\quad\left.\log\log\left(\frac{\sqrt{N}}{\epsilon}L^d\log^{\frac{d}{2}}L\right) \log\left(\frac{N}{\epsilon}L^d\log^{\frac{d}{2}}L\right)\right),
        \label{eq:queryComp}
    \end{align}
    where $\mathcal{U}_{\rm all}:=\{U_{p_i}\}_{i\in[N-1]_0} \cup \{U^{\rm SP}_f\}_f \cup \{U^{\rm amp}_{P_{\vec{l}}}\}_{\vec{l}\in\Lambda_L}$.

    \label{th:main}
\end{theorem}

\ \\

\begin{proof}

\ \\

\noindent \textbf{\underline{Algorithm}}

By combining $U_{p_0}$ and $U^{\rm amp}_{P_{\vec{l}}}$, we get the quantum circuit $O_{1,\vec{l}}$ that acts as
\begin{align}
    & O_{1,\vec{l}}\ket{\vec{x}_0}\ket{0}\ket{0}= \nonumber \\
    & \sum_{\vec{x}_1 \in \mathcal{X}_1} \sqrt{\tilde{p}_0(\vec{x}_1|\vec{x}_0)}  \ket{\vec{x}_0}\ket{\vec{x}_1} \otimes \nonumber \\
    &   \left(\sqrt{\frac{1+P_{\vec{l}}(\vec{x}_1)}{2}}\ket{1}+\sqrt{\frac{1-P_{\vec{l}}(\vec{x}_1)}{2}}\ket{0}\right).
\end{align}
Using this, $\proc{UBQAE}$ gives us an approximation $\hat{b}_{1,\vec{l}}$ of
\begin{equation}
    \tilde{b}_{1,\vec{l}}:= \sum_{\vec{x}_1 \in \mathcal{X}_1} \tilde{p}_0(\vec{x}_1|\vec{x}_0)\frac{1+P_{\vec{l}}\left(\vec{x}_1\right)}{2},
\end{equation}
and thus an approximation $\hat{a}_{1,\vec{l}}$ of
\begin{align}
    \tilde{a}_{1,\vec{l}}&:=(2\tilde{b}_{1,\vec{l}}-1)C(\vec{l}) \nonumber \\
    &=C(\vec{l})\times\sum_{\vec{x}_1 \in \mathcal{X}_1} \tilde{p}_0(\vec{x}_1|\vec{x}_0)P_{\vec{l}}\left(\vec{x}_1\right).
\end{align}
Similarly, for a function $f$ for which we have $U^{\rm SP}_{f}$, combining $U^{\rm SP}_{f}$, $U_{p_i}$, and $U^{\rm amp}_{P_{\vec{l}}}$ gives us the quantum circuit $O_{f,i+1,\vec{l}}$ that acts as
\begin{align}
    & O_{f,i+1,\vec{l}}\ket{0}\ket{0}\ket{0}= \nonumber \\
    & \ \sum_{\vec{x}_i \in \mathcal{X}_{f}} \sum_{\vec{x}_{i+1} \in \mathcal{X}_{i+1}} \sqrt{\tilde{f}(\vec{x}_i)\tilde{p}_i(\vec{x}_{i+1}|\vec{x}_i)}  \ket{\vec{x}_i}\ket{\vec{x}_{i+1}} \nonumber \\
    & \ \otimes \left(\sqrt{\frac{1+P_{\vec{l}}\left(\vec{x}_{i+1}\right)}{2}}\ket{1}+\sqrt{\frac{1-P_{\vec{l}}\left(\vec{x}_{i+1}\right)}{2}}\ket{0}\right),
\end{align}
and using this, $\proc{UBQAE}$ gives us an approximation $\hat{b}_{f,i+1,\vec{l}}$ of
\begin{equation}
    \tilde{b}_{f,i+1,\vec{l}}:=\sum_{\vec{x}_i \in \mathcal{X}_{f}} \sum_{\vec{x}_{i+1} \in \mathcal{X}_{i+1}} \tilde{f}(\vec{x}_i)\tilde{p}_i(\vec{x}_{i+1}|\vec{x}_i)\frac{1+P_{\vec{l}}\left(\vec{x}_{i+1}\right)}{2},
\end{equation}
and thus an approximation $\hat{a}_{f,i+1,\vec{l}}$ of
\begin{align}
    \tilde{a}_{f,i+1,\vec{l}}&:=(2\tilde{b}_{i,\vec{l}}-1)C(\vec{l}) \nonumber \\
    &=C(\vec{l})\times\sum_{\vec{x}_i \in \mathcal{X}_{f}} \sum_{\vec{x}_{i+1} \in \mathcal{X}_{i+1}} \tilde{f}(\vec{x}_i)\tilde{p}_i(\vec{x}_{i+1}|\vec{x}_i)P_{\vec{l}}\left(\vec{x}_{i+1}\right).
\end{align}
Based on these, we construct Algorithm \ref{alg:main} shown below.

\begin{algorithm}[H]
\caption{Get the orthogonal series approximation of the density function of $\vec{X}(t_N)$.}\label{alg:main}
\begin{algorithmic}[1]
\For{$\vec{l}\in\Lambda^\prime_L$}
\State Construct $O_{1,\vec{l}}$.
\State Run $\proc{UBQAE}\left(O_{1,\vec{l}},\epsilon^\prime,\delta^\prime\right)$, where
\begin{align}
    \epsilon^\prime&:=\frac{\epsilon}{4\sqrt{2N}\left(L+\frac{1}{2}\right)^d\left(\log(2L+1)+\frac{1}{2}\right)^{d/2}}, \label{eq:epsPr} \\
    \delta^\prime&:=\frac{\epsilon}{8\sqrt{2}N\left(L+\frac{1}{2}\right)^d\left(\log(2L+1)+\frac{1}{2}\right)^{d/2}}. \label{eq:delPr}
\end{align}
Let the output be $\hat{b}_{1,\vec{l}}$ and define $\hat{a}_{1,\vec{l}}:=(2\hat{b}_{1,\vec{l}}-1)C(\vec{l})$.
\EndFor
\State Define $\hat{f}_1 = \sum_{\vec{l}\in\Lambda_L} \hat{a}_{1,\vec{l}} P_{\vec{l}}$ with $\hat{a}_{1,\vec{0}}=2^{-d}$.
\For{$i=1,\cdots,N-1$}
\For{$\vec{l}\in\Lambda^\prime_L$}
\State Construct $O_{\hat{f}_{i},i+1,\vec{l}}$.
\State Run $\proc{UBQAE}\left(O_{\hat{f}_{i},i+1,\vec{l}},\epsilon^\prime,\delta^\prime\right)$. Let the output be $\hat{b}_{i+1,\vec{l}}$ and define $\hat{a}_{i+1,\vec{l}}:=(2\hat{b}_{i+1,\vec{l}}-1)C(\vec{l})$.
\EndFor
\State Define $\hat{f}_{i+1} = \sum_{\vec{l}\in\Lambda_L} \hat{a}_{i+1,\vec{l}} P_{\vec{l}}$ with $\hat{a}_{i+1,\vec{0}}=2^{-d}$.
\EndFor
\State \Return $\hat{f}_{N}$.
\end{algorithmic}
\end{algorithm}

The rest of the proof is left to Appendix \ref{app:proofMain}.

\end{proof}

To illustrate Algorithm \ref{alg:main} visually, we present a schematic diagram in Figure \ref{fig:OurAlgo}.

Extracting only the dependency on $\epsilon$ and $N$ except logarithmic factors from the bounds \eqref{eq:circuitDepth} and \eqref{eq:queryComp}, we get
\begin{equation}
    \widetilde{O}\left(\frac{\sqrt{N}}{\epsilon}\right)
    \label{eq:circuitDepthSimple}
\end{equation}
and
\begin{equation}
    \widetilde{O}\left(\frac{N^{3/2}}{\epsilon}\right),
    \label{eq:queryCompSimple}
\end{equation}
respectively.
These are the evaluations that have been preannounced in the introduction.

\setcounter{footnote}{3}
\begin{figure}[tp]
%\captionsetup{justification=raggedright}
\begin{center}
\includegraphics[scale=0.40]{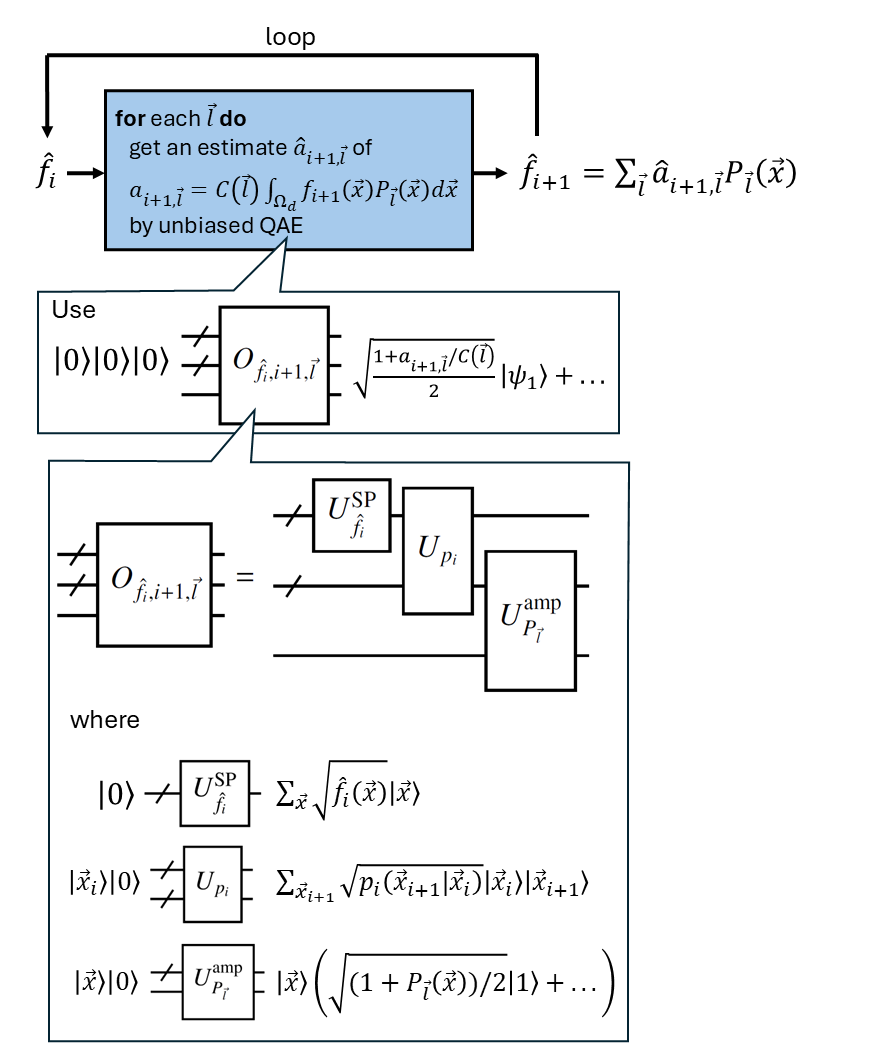}
\caption{Schematic diagram to outline Algorithm \ref{alg:main}. Given an approximating function $\hat{f}_i$ of the density $f_i$, we estimate the coefficients $a_{i+1,\vec{l}}$ by unbiased QAE and get the orthogonal series approximation $\hat{f}_{i+1}$ of the density $f_{i+1}$ at the next time step.
In this QAE, we use an oracle $O_{\hat{f}_{i},i+1,\vec{l}}$ that generates a quantum state encoding $a_{i+1,\vec{l}}$ in the amplitude, and $O_{\hat{f}_{i},i+1,\vec{l}}$ is constructed as a combination of oracles $U^{\rm SP}_{\hat{f}_i}$, $U_{p_i}$, and $U^{\rm amp}_{P_{\vec{l}}}$, which act as Eqs.~\eqref{eq:USPf}, \eqref{eq:Upi}, and \eqref{eq:UPl}, respectively\protect\footnotemark[4]. We iterate this step until we get $\hat{f}_{N}$.}
\label{fig:OurAlgo}
\end{center}
\end{figure}
\footnotetext{Here, for clarity, the resulting states shown in the figure are not exact; see Eqs.~\eqref{eq:USPf}, \eqref{eq:Upi}, and \eqref{eq:UPl} for the exact ones.}

% \section{Discussion \label{sec:disc}}

\section{Comparison of the proposed method with other methods \label{sec:compare}}

Now, let us compare the proposed method with other methods, considering the estimation of expectations at multiple time points: estimating $\mathbb{E}[g_1(\vec{X}(t_1))],\ldots,\mathbb{E}[g_N(\vec{X}(t_N))]$, where $g_1,\ldots,g_N:\Omega_d\rightarrow[0,1]$ are given functions.

First, we consider the naive way that we generate $\ket{f_i}$ by combining $U_{p_0},\ldots,U_{p_{i-1}}$ in series as Fig. \ref{fig:UvecX} and use it in QMCI.
With the accuracy $\epsilon$ required, for estimating the $i$th expectation $\mathbb{E}[g_i(\vec{X}(t_i))]$ in this way, we use a quantum circuit with $\mathcal{U}_p$-depth of order $\widetilde{O}(N/\epsilon)$ and the number of queries to oracles in $\mathcal{U}_p$ is of the same order in total.
For estimating all the expectations, the total query number piles up to $\widetilde{O}(N^2/\epsilon)$.

The proposed method improves this with respect to the scaling on $N$.
We obtain approximating functions $\hat{f}_1,\ldots,\hat{f}_N$ for the density functions $f_1,\ldots,f_N$, using quantum circuits with $\mathcal{U}_{\rm all}$-depth of order $\widetilde{O}(\sqrt{N}/\epsilon)$ and querying the circuits in $\mathcal{U}_{\rm all}$ $\widetilde{O}(N^{3/2}/\epsilon)$ times in total.
After that, we can generate the quantum state $\ket{\hat{f}_i}$ by using $U^{\rm SP}_{\hat{f}_i}$ only once.
If we use this instead of the quantum state $\ket{f_i}$ in QMCI for estimating $\mathbb{E}[g_i(\vec{X}(t_i))]$, the $U^{\rm SP}_{\hat{f}_i}$-depth and the number of queries to $U^{\rm SP}_{\hat{f}_i}$ are $\widetilde{O}(1/\epsilon)$, and the total query number in estimating all the expectations is $\widetilde{O}(N/\epsilon)$.
Throughout the entire process, obtaining $\hat{f}_1,\ldots,\hat{f}_N$ makes the dominant contribution in terms of both the circuit depth and query number, which are, as a consequence, $\widetilde{O}(\sqrt{N}/\epsilon)$ and $\widetilde{O}(N^{3/2}/\epsilon)$, respectively.

Note that there is a version of QMCI that reduces the circuit depth in compensation for the query number \cite{GiurgicaTiron2022lowdepthalgorithms}.
By the algorithm in \cite{GiurgicaTiron2022lowdepthalgorithms}, which is a version of maximum likelihood estimation-based QAE (MLQAE)~\cite{suzuki2020amplitude} , we can estimate $\mathbb{E}[g_i(\vec{X}(t_i))]$ with accuracy $\epsilon$, using quantum circuits with $U_{\vec{X}(t_i)}$-depth of order $\widetilde{O}(1/\epsilon^{1-\beta})$ and querying $U_{\vec{X}(t_i)}$ $\widetilde{O}(1/\epsilon^{1+\beta})$ times in total, where $\beta$ is an arbitrary number in $(0,1]$.
If $U_{\vec{X}(t_i)}$ is a sequence of $O(N)$ oracles in $\mathcal{U}_p$, the $\mathcal{U}_p$-depth and the number of queries to those oracles are $\widetilde{O}(N/\epsilon^{1-\beta})$ and $\widetilde{O}(N/\epsilon^{1+\beta})$, respectively.
This means that, if we set 
\begin{align}
\beta=\log(\sqrt{N})/\log(1/\epsilon),
\label{eq:betaLowDep}
\end{align}
which is in $(0,1]$ as long as $\epsilon<1/\sqrt{N}$, the circuit depth and query number are $\widetilde{O}(\sqrt{N}/\epsilon)$ and $\widetilde{O}(N^{3/2}/\epsilon)$, respectively.
This query number is for estimating a single expectation, and to estimate all the expectations, it piles up to $\widetilde{O}(N^{5/2}/\epsilon)$.

There exist also quantum algorithms that estimate multiple expectations simultaneously \cite{Huggins2022,Cornelissen2022}.
We can apply these algorithms to the current problem: these algorithms can estimate $\mathbb{E}[g_1(\vec{X}(t_1))],\ldots,\mathbb{E}[g_N(\vec{X}(t_N))]$ querying $U_{\vec{X}(t_N)}$, because $U_{\vec{X}(t_N)}$ generates the state that encodes the joint probability distribution of $\vec{X}(t_1),\ldots,\vec{X}(t_N)$ as Eq. \eqref{eq:UXFromUpi}.
In this way, the depth and query number with respect to $U_{\vec{X}(t_N)}$ are both $\widetilde{O}(\sqrt{N}/\epsilon)$, and thus those with respect to oracles in $\mathcal{U}_p$ are both $\widetilde{O}(N^{3/2}/\epsilon)$.

In the classical Monte Carlo method, as discussed in Sec.~\ref{sec:intro}, the total query number scales as $\widetilde{O}(N/\epsilon^2)$.
Here, we consider not queries to $\mathcal{U}_p$ but classical sampling trials from the transition probabilities $p_i$. 

The above discussion is summarized in TABLE \ref{tbl:sum}.
Existing methods can be comparable to the proposed method in terms of either circuit depth or query number, but not in both aspects simultaneously.

\begin{table*}[htbp]
  \centering
  \begin{tabular}{l|c|c}
    Method & Circuit depth & Query number \\ \hline
    Proposed method & $\widetilde{O}(\sqrt{N}/\epsilon)$ & $\widetilde{O}(N^{3/2}/\epsilon)$ \\
    Naive method & $\widetilde{O}(N/\epsilon)$ & $\widetilde{O}(N^2/\epsilon)$ \\
    Low-depth method \cite{GiurgicaTiron2022lowdepthalgorithms} & $\widetilde{O}(\sqrt{N}/\epsilon)$ & $\widetilde{O}(N^{5/2}/\epsilon)$ \\
    Simultaneous method \cite{Huggins2022,Cornelissen2022} & $\widetilde{O}(N^{3/2}/\epsilon)$ & $\widetilde{O}(N^{3/2}/\epsilon)$ \\
    Classical method & --- & $\widetilde{O}(N/\epsilon^2)$
  \end{tabular}
  \caption{Summary of the circuit depth and query number in estimating the expectations $\mathbb{E}[g_1(\vec{X}(t_1))],\ldots,\mathbb{E}[g_N(\vec{X}(t_N))]$ by various methods. Here, we focus on only the scaling on $N$ and $\epsilon$.}
  \label{tbl:sum}
\end{table*}

We should also note the dependency of the complexity of the proposed method on other parameters.
In particular, Eqs. \eqref{eq:circuitDepth} and \eqref{eq:queryComp} have the factors $L^d$ and $L^{2d}$, respectively, which is exponential with respect to the dimension $d$.
These exponential dependencies stem from choosing the tensorized Legendre polynomials as the basis functions for density estimation: the number of the basis functions increases exponentially with respect to $d$.
This may make the proposed method difficult to apply in high-dimensional situations.

\section{Numerical demonstration \label{sec:demo}}

We now conduct a demonstrative experiment to confirm the above theoretical findings numerically.

As an example of stochastic processes that match the current setting that $\vec{X}(t)$ is bounded, $\vec{X}(t)\in\Omega_d$, we take the one-dimensional two-sided reflected Brownian motion $X(t)$~\cite{veestraeten2004conditional,Ivanovs_2010,DAuria2011two}, which is used in fields such as queuing theory and mathematical finance.
It is described by the SDE
\begin{align}
dX(t)=\mu dt + \sigma dW(t)
\end{align}
with constant parameters $\mu\in\mathbb{R}$ and $\sigma\in\mathbb{R}_+$, the initial value $X(t_0)=x_0$, and the reflective boundary condition at $X=c,d$, which are now set to $c=1$ and $d=-1$.
Its transition probability density from $t=s,X(t)=x$ to $t=s^\prime,X(t)=x^\prime$ is analytically given as~\cite{veestraeten2004conditional}
\begin{widetext}
\begin{align}
& p_{\rm RBM}(x^\prime, s^\prime ; x, s) \nonumber \\
&= \sum_{n=-\infty}^{\infty} 
\left\{ 
\frac{1}{\sigma \sqrt{2\pi (s' - s)}}
\exp \left( \frac{2\mu n (c - d)}{\sigma^2} \right) 
\exp \left( \frac{-(x' + 2n(d - c) - x - \mu (s' - s))^2}{2\sigma^2 (s' - s)} \right)
\right\} \nonumber \\
&\quad + \sum_{n=-\infty}^{\infty} 
\left\{ 
\frac{1}{\sigma \sqrt{2\pi (s' - s)}}
\exp \left( \frac{-2\mu (nd - (n+1)c + x)}{\sigma^2} \right) 
\exp \left( \frac{-(2nd - 2(n+1)c + x + x' - \mu (s' - s))^2}{2\sigma^2 (s' - s)} \right)
\right\} \nonumber \\
&\quad - \frac{2\mu}{\sigma^2} 
\sum_{n=0}^{\infty} 
\left\{ 
\exp \left( \frac{2\mu (nd - (n+1)c + x')}{\sigma^2} \right) 
\left( 1 - \Phi \left( \frac{\mu (s' - s) + 2nd - 2(n+1)c + x + x'}{\sigma \sqrt{s' - s}} \right) \right)
\right\} \nonumber \\
&\quad + \frac{2\mu}{\sigma^2} 
\sum_{n=0}^{\infty} 
\left\{ 
\exp \left( \frac{2\mu (nc - (n+1)d + x')}{\sigma^2} \right) 
\Phi \left( \frac{\mu (s' - s) - 2(n+1)d + 2nc + x + x'}{\sigma \sqrt{s' - s}} \right) 
\right\},
\end{align}
\end{widetext}
where $\Phi$ is the cumulative distribution function of the standard normal distribution.
In practice, the infinite sum is truncated at some finite value of $n$, which we denote by $n=\pm n_c$.
Thanks to this formula, we can compare the results of our demonstration with exact values except for the truncation of the infinite sum, and thus evaluate how accurate our method is.
This is also a reason why we choose this stochastic process as a test case.

We run Algorithm \ref{alg:main} for this $X(t)$ to obtain its approximate density functions $\hat{f}_1,\ldots,\hat{f}_N$.
Then, to quantity the accuracy of our method, we estimate
\begin{align}
q_N\coloneqq{\rm Pr}(X(t_N)>x_0)=\int_{x_0}^1 p_{\rm RBM}(x, t_N ; x_0, t_0)dx,    
\end{align}
the probability that $X(t)$ exceeds the initial value at the terminal time, by $\hat{q}_N\coloneqq\int_{x_0}^1 \hat{f}_N(x)dx$, and compare it with the exact value.
Since the current quantum hardware cannot be used for QMCI, we replace it in Algorithm \ref{alg:main} with classical simulation.
That is, we compute
\begin{align}
&b_{\hat{f}_i,i+1,l} := 
\int_{-1}^1 dx_i \int_{-1}^1dx_{i+1} \hat{f}_i(x_i)p_{\rm RBM}(x_{i+1}, t_{i+1};x_i,t_i) \nonumber \\
& \qquad\qquad\qquad\qquad\qquad\qquad\qquad\times \frac{1+P_{l}(x_{i+1})}{2}
\end{align}
classically (concretely, \texttt{dblquad} function in SciPy~\cite{2020SciPy-NMeth}), run a classical simulation of QAE with $b_{\hat{f}_i,i+1,l}$ being the squared amplitude to be estimated, and let its output be $\hat{b}_{i+1,l}$.
Here, instead of the unbiased QAE $\proc{UBQAE}$ in \cite{Cornelissen2023}, we consider the random-depth QAE (RQAE) in \cite{lu2023random}, for which the classical simulation is easier.
As a modified version of MLQAE, RQAE uses not the fixed schedule of the increasing number of Grover operator applications but the randomly fluctuating one, and in \cite{lu2023random} bias reduction was demonstrated numerically despite the lack of a mathematical proof.
We run Algorithm \ref{alg:RQAE}, whose output obeys the same probability distribution as that of RQAE.
Here, $m_{i,j}$ corresponds to the random depth of the quantum circuit.
The total query number in one run of RQAE is counted as $\sum_{i=0}^{N_{\rm r}-1} \sum_{j=1}^R m_{i,j}$, and the maximum depth of quantum circuit used in it is $\max_{i,j} m_{i,j}$.

\begin{algorithm}[H]
\caption{Classical simulation of random-depth QAE\protect\footnotemark[5]}\label{alg:RQAE}
\begin{algorithmic}[1]
\Require $\epsilon$: estimation accuracy, $R$: number of shots in each round, $a$: squared amplitude to be estimated
\State Set $K_\epsilon=\left\lfloor \log_2\left(\frac{1}{\epsilon}\right) \right\rfloor$.
\State Set $N_{\rm r}=K_\epsilon+1$ if $2^{K_\epsilon}<\left\lceil\frac{1}{\epsilon}\right\rceil$ or $N_{\rm r}=K_\epsilon$ otherwise.
\State Set $a_0=a$.
\For{$j=1,\ldots,R$}
\State Generate a sample $n_{0,j}$ from the Bernoulli distribution with $p=a$.
\EndFor
\For{$i=1,\ldots,N_{\rm r}-1$}
    \For{$j=1,\ldots,R$}
        \State Sample an integer $m_{i,j}$ uniformly from $2^i,\ldots,\max\left\{2^{i+1},\left\lceil\frac{1}{\epsilon}\right\rceil\right\}$.
        \State Generate a sample $n_{i,j}$ from the Bernoulli distribution with $p=a_{i,j}\coloneqq\sin^2(m_{i,j}\theta)$, where $\theta=\arcsin \sqrt{a}$.
    \EndFor
\EndFor
\State Find and output $\hat{a}$ that minimizes $\prod_{i=0}^{N_{\rm r}-1} \prod_{j=1}^R \left(a_{i,j}(\hat{a})\right)^{n_{i,j}} (1-a_{i,j}(\hat{a}))^{1-n_{i,j}}$, where $a_{i,j}(\hat{a})=
\sin^2(m_{i,j} \times \arcsin \sqrt{\hat{a}})$ with $m_{0,j}=1$.
\end{algorithmic}
\end{algorithm}
\footnotetext{Here, the original procedure in \cite{lu2023random} is modified in the following points. First, in \cite{lu2023random}, the number $K$ of rounds is an input parameter, but here the accuracy $\epsilon$ is input and $K$ is determined by it. Second, in the last round, the upper bound of the range from which $m_{i,j}$ is drawn is not a power of 2 but $\lceil 1/\epsilon \rceil$, so that we do not make unnecessarily many oracle calls to achieve the accuracy $\epsilon$.}

We compare the results of our method with those of the low-depth method.
We run Algorithm \ref{alg:LowDepthQAE}, a classical simulation of the low-depth QAE (LQAE), inputting $q_N$ as the square amplitude $a$ to be estimated and let its output be the simulated output of the low-depth method for estimating $q_N$.
In Algorithm \ref{alg:LowDepthQAE}, $2m_k+1$ corresponds to the depth of the quantum circuit used in the $k$-th round.
Thus, the total query number in one run of the LQAE is $\sum_{k=1}^K R(2m_k+1)$, and the maximum circuit depth is $2m_K+1$.

\begin{algorithm}[H]
\caption{Classical simulation of low-depth QAE}\label{alg:LowDepthQAE}
\begin{algorithmic}[1]
\Require $\epsilon$: estimation accuracy, $\beta$: speed of the increase of Grover operator applications, $R$: number of shots in each round, $a$: squared amplitude to be estimated
\State Set $K=\left\lceil \max\left\{\epsilon^{-2\beta},\log(1/\epsilon)\right\} \right\rceil$.
\For{$k=1,\ldots,K$}
    \State Generate $R$ samples from the Bernoulli distribution with $p=a_{k}\coloneqq\sin^2((2m_{k}+1)\theta)$, where $m_k\coloneqq\left\lfloor k^{(1-\beta)/2\beta}\right\rfloor$ and $\theta\coloneqq\arcsin \sqrt{a}$, and let the number of 1's be $n_k$.
\EndFor
\State Find and output $\hat{a}$ that minimizes $\prod_{k=1}^{K} \left(a_{k}(\hat{a})\right)^{n_{k}} (1-a_{k}(\hat{a}))^{R-n_{k}}$, where $a_{k}(\hat{a})=\sin^2((2m_{k}+1) \times \arcsin \sqrt{\hat{a}})$.
\end{algorithmic}
\end{algorithm}

\begin{figure}[!htbp]
    \centering
\begin{center}
    \subfigure[RMSE]{
    \includegraphics[width=0.45\textwidth]{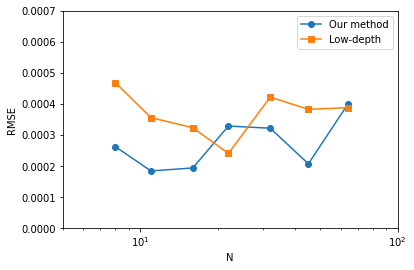} 	\label{fig:RMSE}
    }
    \subfigure[Total query number]{
    \includegraphics[width=0.45\textwidth]{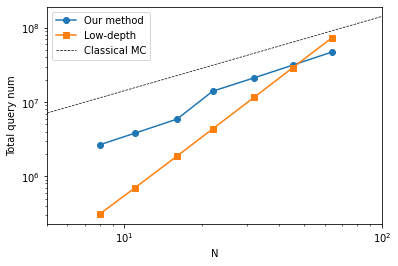} \label{fig:query}
    }
    \subfigure[Maximum circuit depth]{
    \includegraphics[width=0.45\textwidth]{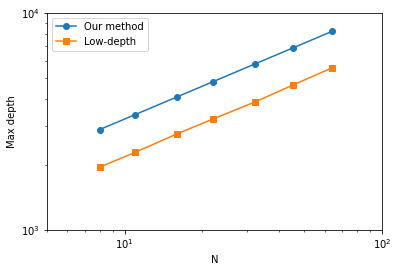} \label{fig:depth}
    }
\caption{The result of the numerical experiments to estimate $q_N$ for various $N$. The blue and orange lines correspond to our method and the low-depth method, respectively. (a) shows the RMSE of the estimates in 10 runs of each method. (b) and (c) show the total query number and the maximum depth of the circuits used in each method averaged over 10 runs, respectively. In (b), we also show the total query number in the classical Monte Carlo method to achieve the RMSE of 0.0004.}
\label{fig:demo}
\end{center}
\end{figure}

In the current demonstration, we do not consider the simultaneous method, which is not formulated as MLQAE and not as easy to simulate as RQAE and LQAE.
We can regard this demonstration as a simulation of QMCI with the circuit depth limited to a certain level, where our method and the low-depth method are natural candidates for the base method, but the simultaneous method is not.

We now show the result of the numerical experiment in FIG.~\ref{fig:demo}.
Here, we set the parameters as follows: $\mu=0.5,\sigma=1,x_0=0,n_c=5,t_0=0$, and $t_1,\ldots,t_N$ are $N$ equidistant points between $t_1=0.2$ and $t_N=0.6$.
We vary $N$ as $N=8,11,16,22,32,45,64$ and see how the result changes.
In our method, we set $R=12$ and $\epsilon=\epsilon^{\rm RQAE}_N\coloneqq2^{-10}/\sqrt{N}$ in Algorithm \ref{alg:RQAE} and the maximum polynomial degree in Legendre expansion to $L=5$.
In LQAE as Algorithm \ref{alg:LowDepthQAE}, we set $R=12$, $\epsilon=0.0029$, and $\beta$ by Eq.~\eqref{eq:betaLowDep}.
FIG.~\ref{fig:RMSE} shows the root-mean-square errors (RMSEs) of the estimates in 10 runs of our method and the low-depth method, which indicates that in the current setting, our method has a better or comparable accuracy compared to the low-depth method for every $N$.
FIG.~\ref{fig:query} shows the total query number in each method averaged over 10 runs.
It also shows, as a reference, the number of sampling trials from the transition probabilities in the classical Monte Carlo method with an RMSE of 0.0004, below which our method's RMSE is for every $N$.
From this figure, we see that, as expected, the low-depth method has a rapidly increasing query number of $\widetilde{O}(N^{5/2})$.
On the other hand, the query number in our method shows the milder increase, which is of $\widetilde{O}(N^{3/2})$\footnote{In fact, the total query number in our method non-continuously increases with respect to $N$, like a bump between $N=16$ and $22$ in FIG.~\ref{fig:query}. This is understood as follows. In Algorithm~\ref{alg:RQAE}, the considerable contribution to the total query number comes from the second-to-last round, $i=N_{\rm r}-2$, where the upper bound of the random depth $m_{i,j}$ is $2^{K_\epsilon}$. With $\epsilon$ set to $\epsilon^{\rm RQAE}_N$ in the current demonstration, $2^{K_\epsilon}$ non-continuously increases with respect to $N$, which leads to the non-continuous increase of the total query number. Nevertheless, the total query number scales as $\widetilde{O}(N^{3/2})$.}, and is smaller than that of the low-depth method for $N=64$.
FIG.~\ref{fig:depth} shows the maximum depth of the circuits used in each method averaged over 10 runs.
Both our method and the low-depth method show the scaling of $\widetilde{O}(\sqrt{N})$ as expected.

\section{Summary \label{sec:sum}}

In this paper, we have focused on the time evolution of a stochastic process $\vec{X}(t)$ in QMCI for estimating expectations concerning $\vec{X}(t)$.
Given the transition probability densities $p_0,\ldots,p_{N-1}$ between discrete time points $t_0,\ldots,t_N$, we can combine the quantum circuits $U_{p_0},\ldots,U_{p_{N-1}}$ for state preparation for $\{p_i\}_i$ to generate a state encoding the distribution of $\vec{X}(t_N)$.
However, this results in $O(N)$ circuit depth, and if we estimate expectations concerning $\vec{X}(t_1),\ldots,\vec{X}(t_N)$, the total number of queries to $U_{p_i}$ scales on $N$ as $O(N^2)$.
We have aimed to improve this.
We have proposed a method to divide the time evolution quantum circuits based on OSDE.
In this method, we estimate the coefficients in the orthogonal series approximation of the density $f_i$ of $X(t_i)$ by unbiased QAE-based QMCI and get the approximating function $\hat{f}_i$ of $f_i$, and this $\hat{f}_i$ is used in estimating the next density $f_{i+1}$.
We iterate this step and get all the approximations $\hat{f}_1,\ldots,\hat{f}_N$, which can be used to estimate expectations concerning $\vec{X}(t_1),\ldots,\vec{X}(t_N)$.
In this approach, the dependencies of the circuit depth and the total query number on $N$ are reduced to $O(\sqrt{N})$ and $O(N^{3/2})$, respectively.
We have also compared our method with other versions of QMCI that aim for depth reduction or estimating multiple expectations, as summarized in Table \ref{tbl:sum}.
Besides, we have conducted a numerical demonstration of our method on the simulation of the reflected Brownian motion and seen that our method can have an advantage in terms of the query number compared to the low-depth method.

We should note that the current setting does not necessarily match some important problems.
In particular, as described in Appendix \ref{app:proofMain}, we make assumptions of the boundedness of $\vec{X}(t)$ and the positive lower bound of the density in order for the estimated density to be bona fide, and this does not match cases that $\vec{X}(t)$ is unbounded as is common in derivative pricing.
We will try to extend our algorithm so that it can be applied to more general settings in future works.

Although in this paper we consider a problem setting with derivative pricing as an application example in mind, it is also interesting to explore further applications beyond derivative pricing, using the currently proposed method as a base technology.
A possible target is the McKean-Vlasov process~\cite{McKean1966} described by the SDE
\begin{align}
    d\vec{X}(t)=\vec{\mu}\left(t,\vec{X}(t),f_{\vec{X}(t)}\right)dt+\Sigma\left(t,\vec{X}(t),f_{\vec{X}(t)}\right)d\vec{W}(t)
\end{align}
with drift and diffusion coefficients depending on $\vec{X}(t)$'s distribution $f_{\vec{X}(t)}$ in the form of expectation, which often appears not only in finance~\cite{guyon2013nonlinear} but also in fields of physics such as fluid dynamics~\cite{bossy1997stochastic} and interacting particle systems~\cite{Meleard1996}.
Solving this SDE involves estimating $\vec{\mu}$ and $\Sigma$ as expectations at every time step, to which we expect the currently proposed method can be applied.
Such an application will be explored in future works.

\section*{Acknowledgements}

The author is supported by MEXT Quantum Leap Flagship Program (MEXT Q-LEAP) Grant no. JPMXS0120319794, JSPS KAKENHI Grant no. JP22K11924, and JST COI-NEXT Program Grant No. JPMJPF2014.

\appendix

\section{Exact statement of Theorem \ref{th:LegExAccu} \label{app:LegExAccu}}

\begin{definition}
    For $s,a\in\mathbb{R}_+$, $N_{s,a}\subset\mathbb{C}$ denotes the open region bounded by the ellipse with foci 0 and $s$, and leftmost point $-a$. 
\end{definition}

\begin{definition}
    For $h\in\mathbb{R}_+$, we define
    \begin{equation}
        D_{h} := \left\{(z_1,\cdots,z_d)\in\mathbb{C}^d \ \middle| \ \sum_{i=1}^d z_i^2 \in N_{d,h^2}\right\}.
    \end{equation}
\end{definition}

\begin{theorem1restate}
Let $f:\Omega_d\rightarrow\mathbb{R}$ be an analytic function and suppose that there exists $h\in\mathbb{R}_+$ such that $f$ has an analytic extension to $D_h$.
Then, there exists $K\in\mathbb{R}_+$ and $\rho\in\mathbb{R}_{>1}$ such that, for any integer $L$ satisfying $L>\frac{d}{2\log \rho}$,
\begin{equation}
    \max_{\vec{x}\in\Omega_d} \left|\mathcal{P}_L[f](\vec{x})-f(\vec{x})\right| \le K\rho^{-L}
    \label{eq:LegExAccu}
\end{equation}
holds.
\end{theorem1restate}

\section{More on Theorem \ref{th:main} \label{app:proofMain}}

\subsection{Exact statement and proof of Theorem \ref{th:main}}

\begin{theoremmainrestate}
    Let $\alpha,\delta,\epsilon\in(0,1)$.
    Let $\{\vec{X}(t_{i})\}_{i\in[N]_0}$ be a $\Omega_d$-valued stochastic process with the deterministic initial value $\vec{x}_0$ and the conditional transition probability density $p_i(\cdot|\cdot)$.
    Suppose that the following assumptions on the access to oracles hold:

    \begin{enumerate}
    \renewcommand{\labelenumi}{[AS\arabic{enumi}]}
        \item For each $i\in[N]$, $\vec{X}(t_i)$ has the density function $f_i$ that satisfies the condition in Theorem \ref{th:LegExAccu} with common $h\in\mathbb{R}_+$, Eq. \eqref{eq:LegExAccu} with common $K\in\mathbb{R}_+$ and $\rho\in\mathbb{R}_{>1}$, and
        \begin{equation}
            \min_{\vec{x}\in\Omega_d} f_i(\vec{x}_i) \ge f_{\rm min} := \frac{2(L+1)^{3d/2}\epsilon\sqrt{N}}{\sqrt{\alpha}\left(\log(2L+1)+\frac{1}{2}\right)^{d/2}},
            \label{eq:fmin}
        \end{equation}
        where
        \begin{equation}
            L:=\frac{1}{\log \rho}\log\left(\frac{2^{(d+2)/2}K}{\epsilon}\right).
            \label{eq:L}
        \end{equation}

        \item 
        For each $i\in[N-1]_0$, we have access to the oracle $U_{p_i}$ acting as Eq. \eqref{eq:Upi}.
        Here, $\mathcal{X}_1,\ldots,\mathcal{X}_N\subset\Omega_d$ are finite sets, $\mathcal{X}_0:=\{\vec{x}_0\}$, and $\tilde{p}_i:\mathcal{X}_{i+1}\times\mathcal{X}_i\rightarrow[0,1],i\in[N-1]_0$ satisfies $\sum_{\vec{x}_{i+1}\in\mathcal{X}_{i+1}}\tilde{p}_i(\vec{x}_{i+1}|\vec{x}_i)=1$ for any $\vec{x}_i\in\mathcal{X}_i$.

        \item  For any $\vec{l}\in\Lambda_L$,
        \begin{align}
            &\left|\int_{\Omega_d}  f_1(\vec{x}_{1})P_{\vec{l}}(\vec{x}_{1}) d\vec{x}_1 -\sum_{\vec{x}_{1} \in \mathcal{X}_1} \tilde{p}_0(\vec{x}_{1}|\vec{x}_0)P_{\vec{l}}(\vec{x}_{1})\right| \nonumber \\
            \le & \frac{\epsilon}{4\sqrt{2}N\left(\log(2L+1)+\frac{1}{2}\right)^{d/2}}
            \label{eq:p0FiniteApp}
        \end{align}
        holds.

        \item For $f:\Omega_d\rightarrow\mathbb{R}_{\ge 0}$ written as $f=\sum_{\vec{l}\in\Lambda_L} a_{\vec{l}} P_{\vec{l}}$ with any real number set $\{a_{\vec{l}}\}_{\vec{l}\in\Lambda_L}$ including $a_{\vec{0}}=2^{-d}$, we have access to the oracle $U^{\rm SP}_{f}$ that acts as Eq.~\eqref{eq:USPf}.
        Here, $\mathcal{X}_{f} \subset \Omega_d$ is a finite set, and $\tilde{f}:\mathcal{X}_{f}\rightarrow[0,1]$ is a map such that $\sum_{\vec{x}\in\mathcal{X}_{f}}\tilde{f}(\vec{x})=1$.
        Besides, $\mathcal{X}_{f}$ and $\tilde{f}$, along with $\mathcal{X}_i$ and $\tilde{p}_i$ in (A2), satisfy
        \begin{align}
            &\left|\int_{\Omega_d} d\vec{x}_i \int_{\Omega_d} d\vec{x}_{i+1} f(\vec{x}_i)p_i(\vec{x}_{i+1}|\vec{x}_i)P_{\vec{l}}(\vec{x}_{i+1})-\right. \nonumber \\
            & \quad \left.\sum_{\vec{x}_i \in \mathcal{X}_{f}}\sum_{\vec{x}_{i+1} \in \mathcal{X}_{i+1}} \tilde{f}(\vec{x}_i)\tilde{p}_i(\vec{x}_{i+1}|\vec{x}_i)P_{\vec{l}}(\vec{x}_{i+1})\right| \nonumber \\
            \le & \frac{\epsilon}{4\sqrt{2}N\left(\log(2L+1)+\frac{1}{2}\right)^{d/2}}
            \label{eq:piFiniteApp}
        \end{align}
        for any $\vec{l}\in\Lambda_L$ and $i\in[N-1]$.

        \item For any $\vec{l}\in\mathbb{N}_0^d$, we have access to the oracle $U^{\rm amp}_{P_{\vec{l}}}$ that acts as Eq.~\eqref{eq:UPl} for any $\vec{x}\in\mathbb{R}^d$.

        \item For any $i\in[N-1]_0$ and  $\vec{l}\in\Lambda_L$,
        \begin{equation}
            \sum_{\vec{l}^\prime\in\Lambda_L^\prime}|c_{i,\vec{l},\vec{l}^\prime}|\le1,
            \label{eq:csum}
        \end{equation}
        holds, where
        \begin{equation}
        c_{i,\vec{l},\vec{l}^\prime}:=C(\vec{l})\int_{\Omega_d} d\vec{x}_i \int_{\Omega_d} d\vec{x}_{i+1} P_{\vec{l}^\prime}(\vec{x}_{i})p_i(\vec{x}_{i+1}|\vec{x}_i)P_{\vec{l}}(\vec{x}_{i+1}).
        \end{equation}

    \end{enumerate}
    Then, there exists a quantum algorithm that, with probability at least $1-\alpha$, outputs approximation functions $\hat{f}_1,\ldots,\hat{f}_N:\Omega_d\rightarrow\mathbb{R}_{\ge0}$ for $f_1,\ldots,f_N$ with the following properties: each $\hat{f}_i$ is written as $\hat{f}_i=\sum_{\vec{l}\in\Lambda_L} \hat{a}_{i,\vec{l}}P_{\vec{l}}$ with random coefficients $\{\hat{a}_{i,\vec{l}}\}_{\vec{l}\in\Lambda^\prime_L}$ and $\hat{a}_{i,\vec{0}}=2^{-d}$, and its MISE is bounded as Eq.~\eqref{eq:MISE}.
    In this algorithm, we use quantum circuits with $\mathcal{U}_{\rm all}$-depth of order \eqref{eq:circuitDepth}, and the total number of queries to oracles in $\mathcal{U}_{\rm all}$ is of order \eqref{eq:queryComp},  where $\mathcal{U}_{\rm all}:=\{U_{p_i}\}_{i\in[N-1]_0} \cup \{U^{\rm SP}_f\}_f \cup \{U^{\rm amp}_{P_{\vec{l}}}\}_{\vec{l}\in\Lambda_L}$.

\end{theoremmainrestate}

\begin{proof}
The algorithm is as shown in Sec.~\ref{sec:main}.
The rest of the proof is as follows.
\ \\

\noindent \textbf{\underline{Accuracy}}

We define
\begin{align}
    \delta f_i & := \hat{f}_i - \mathcal{P}_L[f_i] := \sum_{\vec{l}\in\Lambda_L^\prime} \delta a_{i,\vec{l}} P_{\vec{l}},
\end{align}
where
\begin{align}
    \delta a_{i,\vec{l}} := \hat{a}_{i,\vec{l}} - a_{f_i,\vec{l}}.
\end{align}

Now, we show that for any $i\in[N]$, we have a decomposition that
\begin{align}
&\delta a_{i,\vec{l}}=\delta a^{\rm QAE}_{i,\vec{l}}+\delta a^{\rm disc}_{i,\vec{l}}+ \nonumber \\
&\sum_{j=1}^{i-1} \sum_{\vec{l}_{i-1}\in\Lambda_L^\prime} \cdots 
\sum_{\vec{l}_{j}\in\Lambda_L^\prime} c_{i-1,\vec{l},\vec{l}_{i-1}} c_{i-2,\vec{l}_{i-1},\vec{l}_{i-2}} \cdots c_{j,\vec{l}_{j+1},\vec{l}_j} \nonumber \\
&\qquad\qquad\qquad\qquad\quad\times\left(\delta a^{\rm QAE}_{j,\vec{l}_j}+\delta a^{\rm disc}_{j,\vec{l}_j}\right)
\label{eq:deltaaDecomp}
\end{align}
with
\begin{align}
    \delta a^{\rm QAE}_{i,\vec{l}} & :=
    \begin{cases}
    \hat{a}_{1,\vec{l}} - \tilde{a}_{1,\vec{l}} & ; \ i=1 \\
    \hat{a}_{i,\vec{l}} - \tilde{a}_{\hat{f}_{i-1},i,\vec{l}} & ; \ i \ge 2
    \end{cases}
    ,
    \\
    \delta a^{\rm disc}_{i,\vec{l}} & :=
    \begin{cases}
    \tilde{a}_{1,\vec{l}} - a_{f_1,\vec{l}} & ; \ i=1 \\
    \tilde{a}_{\hat{f}_{i-1},i,\vec{l}} - a_{\hat{f}_{i-1},i,\vec{l}}  & ; \ i \ge 2
    \end{cases}
    .
\end{align}
For $i=1$, this holds trivially.
Suppose that this holds for $i$.
We have
\begin{equation}
    \delta a_{i+1,\vec{l}}=\delta a^{\rm QAE}_{i+1,\vec{l}}+\delta a^{\rm disc}_{i+1,\vec{l}}+\delta a^{\rm acc}_{i+1,\vec{l}},
\end{equation}
where $\delta a^{\rm acc}_{i+1,\vec{l}}$ is defined as $\delta a^{\rm acc}_{i+1,\vec{l}} := a_{\hat{f}_i,i+1,\vec{l}} - a_{f_{i+1},\vec{l}}$ with
\begin{align}
    & a_{\hat{f}_i,i+1,\vec{l}} := \nonumber \\
    & \ C(\vec{l}) \times \int_{\Omega_d} d\vec{x}_i \int_{\Omega_d} d\vec{x}_{i+1} \hat{f}_i(\vec{x}_i)p_i(\vec{x}_{i+1}|\vec{x}_i)P_{\vec{l}}(\vec{x}_{i+1})
\end{align}
and
\begin{align}
    & a_{f_{i+1},\vec{l}} = \nonumber \\
    & \ C(\vec{l}) \times \int_{\Omega_d} d\vec{x}_i \int_{\Omega_d} d\vec{x}_{i+1} f_i(\vec{x}_i)p_i(\vec{x}_{i+1}|\vec{x}_i)P_{\vec{l}}(\vec{x}_{i+1})
\end{align}
and thus written as
\begin{align}
    \delta a^{\rm acc}_{i+1,\vec{l}} = \sum_{\vec{l}_i\in\Lambda_L^\prime}c_{i,\vec{l},\vec{l}_i}\delta a_{i,\vec{l}_i}.
\end{align}
By plugging Eq. \eqref{eq:deltaaDecomp} for $i$ into this, we obtain Eq. \eqref{eq:deltaaDecomp} for $i+1$.

Then, let us evaluate the variance of $\delta a_{i,\vec{l}}$.
To do so, we first see that due to the assumption [AS6],
\begin{align}
    & \left|\sum_{\vec{l}_{i-1}\in\Lambda_L^\prime} \cdots 
\sum_{\vec{l}_{j}\in\Lambda_L^\prime} c_{i-1,\vec{l},\vec{l}_{i-1}} c_{i-2,\vec{l}_{i-1},\vec{l}_{i-2}} \cdots c_{j,\vec{l}_{j+1},\vec{l}_j}\right| \nonumber \\
\le & \sum_{\vec{l}_{i-1}\in\Lambda_L^\prime} \cdots 
\sum_{\vec{l}_{j}\in\Lambda_L^\prime} \left|c_{i-1,\vec{l},\vec{l}_{i-1}}\right|\cdot\left|c_{i-2,\vec{l}_{i-1},\vec{l}_{i-2}}\right|\cdot...\cdot\left|c_{j,\vec{l}_{j+1},\vec{l}_j}\right| \nonumber \\
\le & 1
\end{align}
holds.
Besides,
\begin{align}
    &\left|\mathbb{E}_{\rm Q}\left[\delta a^{\rm QAE}_{i,\vec{l}}\right]\right|\le 2C(\vec{l})\delta^\prime, \\
    &\mathbb{E}_{\rm Q}\left[\left(\delta a^{\rm QAE}_{i,\vec{l}}\right)^2\right] \le \left(2C(\vec{l})\epsilon^{\prime}\right)^2
\end{align}
and
\begin{equation}
    |\delta a^{\rm disc}_{i,\vec{l}}| \le 2C(\vec{l})\delta^{\prime}
\end{equation}
follow from Theorem \ref{th:UBQAE} and the assumption [AS4]. 
Using these and noting that $\left\{\delta a^{\rm QAE}_{i,\vec{l}}\right\}_{i,\vec{l}}$ are mutually independent, we have
\begin{align}
    \left|\mathbb{E}_{\rm Q}\left[\delta a_{i,\vec{l}} \right]\right|&\le i \times \left(2C(\vec{l})\delta^{\prime} + 2C(\vec{l})\delta^{\prime}\right) \nonumber \\
    &\le \frac{\epsilon}{2\sqrt{2}\left(\log(2L+1)+\frac{1}{2}\right)^{d/2}} 
    \label{eq:coefEV}
\end{align}
and
\begin{align}
    &\mathbb{E}_{\rm Q}\left[\left(\delta a_{i,\vec{l}} \right)^2\right] \nonumber \\
    \le & i \times \left(2C(\vec{l})\epsilon^{\prime}\right)^2 + i^2 \times \left(2C(\vec{l})\delta^{\prime}\right)^2 + \nonumber \\
    & 2i^2 \times 2C(\vec{l})\delta^{\prime} \times 2C(\vec{l})\delta^{\prime} + i^2 \times \left(2C(\vec{l})\delta^{\prime}\right)^2 \nonumber \\
    \le & \frac{\epsilon^2}{4\left(\log(2L+1)+\frac{1}{2}\right)^d}.
    \label{eq:coefVar}
\end{align}

Finally, let us evaluate the MISE of $\hat{f}_i$.
Young's inequality gives us
\begin{align}
    & \mathbb{E}_{\rm Q}\left[\int_{\Omega_d} \left(\hat{f}_i(\vec{x})-f_i(\vec{x})\right)^2 d\vec{x}\right] \nonumber \\
    \le & 2\mathbb{E}_{\rm Q}\left[\int_{\Omega_d} \left(\delta f_i(\vec{x})\right)^2 d\vec{x}\right] + 2 \int_{\Omega_d} \left(\mathcal{P}_L[f_i](\vec{x})-f_i(\vec{x})\right)^2 d\vec{x}.
    \label{eq:MISEYoung}
\end{align}
We have
\begin{align}
    \mathbb{E}_{\rm Q}\left[\int_{\Omega_d} \left(\delta f_i(\vec{x})\right)^2 d\vec{x}\right] \le \sum_{\vec{l}\in\Lambda_L^\prime} \frac{1}{C(\vec{l})}\times \mathbb{E}_{\rm Q}\left[\left(\delta a_{i,\vec{l}} \right)^2\right] \le \frac{\epsilon^2}{4},
    \label{eq:MISE1stTerm}
\end{align}
where we use Eqs. \eqref{eq:Orthddim} and \eqref{eq:coefVar} along with
\begin{equation}
    \sum_{\vec{l}\in\Lambda_L^\prime} \frac{1}{C(\vec{l})}\le \left(\log(2L+1)+\frac{1}{2}\right)^d,
\end{equation}
which can be checked by elementary calculus.
Under the definition of $L$ in Eq. \eqref{eq:L}, Theorem \ref{th:LegExAccu} implies
\begin{equation}
    \max_{\vec{x}\in\Omega_d} \left|\mathcal{P}_L[f](\vec{x})-f(\vec{x})\right| \le \frac{\epsilon}{2^{(d+2)/2}}.
    \label{eq:MISE2ndTerm}
\end{equation}
Plugging Eqs. \eqref{eq:MISE1stTerm} and \eqref{eq:MISE2ndTerm} into Eq. \eqref{eq:MISEYoung} leads to Eq. \eqref{eq:MISE}.

\ \\

\noindent \textbf{\underline{Positive definiteness}}

Here, we show that, with probability at least $1-\alpha$, $\hat{f}_1,\ldots,\hat{f}_N$ output by Algorithm \ref{alg:main} are positive definite.
We have
\begin{align}
    & {\rm Pr}\left\{\left|\delta a_{i,\vec{l}}\right|\ge\frac{f_{\rm min}}{2(L+1)^d}\right\} \nonumber \\
    \le & {\rm Pr}\left\{\left|\delta a_{i,\vec{l}}-\mathbb{E}_{\rm Q}\left[\delta a_{i,\vec{l}}\right]\right| \ge \frac{f_{\rm min}}{2(L+1)^d}-\left|\mathbb{E}_{\rm Q}\left[\delta a_{i,\vec{l}}\right]\right|\right\} \nonumber \\
    \le & {\rm Pr}\left\{\left|\delta a_{i,\vec{l}}-\mathbb{E}_{\rm Q}\left[\delta a_{i,\vec{l}}\right]\right| \ge \frac{(L+1)^{d/2}\epsilon \sqrt{N}}{2\sqrt{\alpha}\left(\log(2L+1)+\frac{1}{2}\right)^{d/2}}\right\} \nonumber \\
    \le & \frac{\alpha}{N(L+1)^d},
    \label{eq:deltaaTail}
\end{align}
where we use Eqs. \eqref{eq:fmin} and \eqref{eq:coefEV} at the second inequality, and the third inequality follows from Chebyshev's inequality and Eq. \eqref{eq:coefVar}.
Note that if $\left|\delta a_{i,\vec{l}}\right|<\frac{f_{\rm min}}{2(L+1)^d}$ holds for any $i\in[N]$ and $\vec{l}\in\Lambda_L$, we have
\begin{align}
    & \hat{f}_i(\vec{x}) \nonumber \\
    = & f_i(\vec{x}) + \left(\hat{f}_i(\vec{x}) - \mathcal{P}_L[f_i]\right) + \left(\mathcal{P}_L[f_i](\vec{x})-f_i(\vec{x})\right) \nonumber \\
    \ge& f_{\rm min} - \sum_{\vec{l}\in\Lambda_L^\prime} \left|\delta a_{i,\vec{l}}\right|\cdot\left|P_{\vec{l}}(\vec{x})\right| - \frac{\epsilon}{2^{(d+2)/2}} \nonumber \\   
    \ge & 0
\end{align}
for any $i\in[N]$ and $\vec{x}\in\Omega_d$, where we use Eq. \eqref{eq:MISE2ndTerm} at the first inequality and we use $\left|\delta a_{i,\vec{l}}\right|<\frac{f_{\rm min}}{2(L+1)^d}$, $\left|P_{\vec{l}}(\vec{x})\right|\le1$, and Eq. \eqref{eq:fmin} at the second inequality.
Thus, we obtain
\begin{align}
    & {\rm Pr}\left\{{\rm All \ of \ } \hat{f}_1,\ldots,\hat{f}_N \ {\rm are \ positive \ definite}\right\} \nonumber \\
    \ge & {\rm Pr}\left\{\left|\delta a_{i,\vec{l}}\right|<\frac{f_{\rm min}}{2(L+1)^d} \ {\rm for \ all \ } i\in[N],\vec{l}\in\Lambda_L\right\}  \nonumber \\
    \ge & 1 - \sum_{i\in[N]}\sum_{\vec{l}\in\Lambda_L^\prime} {\rm Pr}\left\{\left|\delta a_{i,\vec{l}}\right|\ge\frac{f_{\rm min}}{2(L+1)^d}\right\} \nonumber \\
    \ge & 1 - N\times(L+1)^d\times\frac{\alpha}{N(L+1)^d} \nonumber \\
    \ge & 1-\alpha,
\end{align}
where we use Eq. \eqref{eq:deltaaTail} at the third inequality.

\ \\

\noindent \textbf{\underline{Complexity}}

Lastly, let us evaluate the circuit depth and the total query complexity of Algorithm \ref{alg:main}, in terms of the number of queries to the oracles $U_{p_i}$, $U^{\rm SP}_f$, and $U^{\rm amp}_{P_{\vec{l}}}$.
Since each $\proc{UBQAE}$ in Algorithm \ref{alg:main} is run as a separate quantum algorithm, the number of oracle calls in it is an upper bound of the depth of quantum circuits used in Algorithm \ref{alg:main}.
Because we make $O(1)$ uses of $U_{p_i}$, $U^{\rm SP}_f$, and $U^{\rm amp}_{P_{\vec{l}}}$ in $O_{1,\vec{l}}$ and $O_{\hat{f}_{i},i+1,\vec{l}}$, we see from Theorem \ref{th:UBQAE} that each $\proc{UBQAE}$ calls these oracles $O\left(\frac{1}{\epsilon^\prime}\log\log\left(\frac{1}{\epsilon^\prime}\right)\log\left(\frac{1}{\delta^\prime\epsilon^\prime}\right)\right)$ times.
Plugging Eqs. \eqref{eq:epsPr} and \eqref{eq:delPr} into this followed by some simplification yields the circuit depth bound in Eq. \eqref{eq:circuitDepth}.
Because Algorithm \ref{alg:main} runs $\proc{UBQAE}$ $N(L+1)^d$ times in total, multiplying this to Eq. \eqref{eq:circuitDepth} yields the total query number bound in Eq. \eqref{eq:queryComp}.
\end{proof}

\subsection{On the assumptions in Theorem \ref{th:main}}

Theorem \ref{th:main} accompanies many assumptions and for some of them, it is not apparent why they are made and whether they are reasonable.
We now explain such points.

In [AS1], we assume that the density functions $f_i$ are lower bounded by some positive number as Eq. \eqref{eq:fmin}.
We make this assumption in order to ensure that the estimated density functions $\hat{f}_i$ are bona fide \cite{Efromovich2010}.
Because of the errors due to Legendre expansion and estimating the coefficients by QAE, the value of $\hat{f}_i$ at each point in $\Omega_d$ deviates from the true value and can be even negative.
However, if $f_i$ has a positive lower bound and the error in $\hat{f}_i$ is suppressed compared to the bound, it is ensured that $\hat{f}_i$ is positive definite.
Note that $\int_{\Omega_d} \hat{f}_i(\vec{x}) d\vec{x}$ is satisfied because of setting $\hat{a}_{i,\vec{0}}=2^{-d}$.

Assuming that each $f_i$ has the positive lower bound is reasonable in the current setting that $\vec{X}(t)$ is in $\Omega_d$, a bounded region.
There are some concrete examples of stochastic processes that have bounded ranges and are of practical interest such as the two-sided reflected Brownian motion considered in Sec.~\ref{sec:demo}.
This makes the proposed method with the assumption of the bounded range meaningful. 
Nevertheless, many stochastic processes such as Brownian motions with no reflective boundary have unbounded ranges, and their density functions can be arbitrarily close to 0.
Extending the proposed method to such a case with estimated density functions kept bona fide is left for future work.

[AS2], [AS3], and [AS4] are on the availability of the state preparation oracles for the transition probability densities $p_i$ and Legendre series.
As explained in Sec. \ref{sec:QMCI}, there are some methods for state preparation with exponentially fine grid points, and thus assuming that integrals can be approximated by finite sums on the grid points as Eqs. \eqref{eq:p0FiniteApp} and \eqref{eq:piFiniteApp} is reasonable.

The oracles in [AS5] can be also implemented with arithmetic circuits and controlled rotation gates, as mentioned in Sec. \ref{sec:QMCI}.

[AS6] looks less trivial than the others.
To consider this, we expand $p_i(\vec{x}_{i+1}|\vec{x}_i)$ as a function of $\vec{x}_i$ and $\vec{x}_{i+1}$ with tensorized Legendre polynomials:
\begin{equation}
    p_i(\vec{x}_{i+1}|\vec{x}_i) = \sum_{\vec{l},\vec{l}^\prime\in \mathbb{N}_0^d} a_{p_i,\vec{l},\vec{l}^\prime} P_{\vec{l}}(\vec{x}_{i+1}) P_{\vec{l}^\prime}(\vec{x}_i). 
\end{equation}
Then, we see that Eq. \eqref{eq:csum} holds if
\begin{equation}
    \sum_{\vec{l}^\prime\in\Lambda_L^\prime}\frac{\left|a_{p_i,\vec{l},\vec{l}^\prime}\right|}{C(\vec{l}^\prime)}\le1.
    \label{eq:csum2}
\end{equation}
Even if $\left|a_{p_i,\vec{l},\vec{l}^\prime}\right|$ decreases with respect to $\vec{l}^\prime$ weakly, say $\left|a_{p_i,\vec{l},\vec{l}^\prime}\right| = O\left(\left(\|\vec{l}^\prime\|_\infty\right)^{-\omega}\right)$ with $\omega\in\mathbb{R}_+$, the LHS of Eq. \eqref{eq:csum2} converges with small $\vec{l}^\prime$ making the dominant contribution.
This consideration implies that [AS6] holds in a wide range of situations.

\bibliographystyle{apsrev4-1}
\bibliography{reference}

\end{document}